\documentclass[11pt]{article}
\usepackage[shortlabels]{enumitem}
\usepackage{fullpage}
\usepackage{comment}
\usepackage{amsmath}
\usepackage{amsthm}
\usepackage{amssymb}
\usepackage{authblk}
\usepackage{hyperref}
\usepackage{color}
\usepackage{url}
\usepackage[inkscape={/Applications/Inkscape.app/Contents/Resources/bin/inkscap??e -z -C}]{svg}
\usepackage[final]{showkeys}
\usepackage{cleveref}
\usepackage{tikz}
\usepackage{lineno}

\usepackage{microtype}
\usepackage{mathtools}
\usepackage{tikz}
\usepackage{cleveref}
\usepackage{graphicx, import}
\usepackage{grffile}

\usetikzlibrary{arrows.meta}

\usepackage{algorithmic}
\usepackage{floatpag}
\usepackage[ruled,vlined,commentsnumbered,titlenotnumbered]{algorithm2e}

\usepackage[toc,page]{appendix}

\newtheorem{theorem}{Theorem}
\newtheorem{lemma}{Lemma}

\newtheorem{claim}{Claim}
\newtheorem{corollary}{Corollary}
\newtheorem{definition}{Definition}

\newcommand{\g}{G}

\newcommand{\sg}[1]{G}
\newcommand{\p}{\mathsf{Peb}}
\newcommand{\ti}{\mathsf{Time}}
\newcommand{\s}{\mathcal{P}}

\newcommand{\hf}{\mathbb{H}}
\newcommand{\shf}[1]{\mathbb{H}_{#1}}

\newcommand{\sh}[1]{H_{n, #1}}
\newcommand{\shn}[2]{H_{#2, #1}}
\newcommand{\ai}{\{a_i^1, \dots, a_i^K\}}
\newcommand{\bi}{\{b_i^1, \dots, b_i^K\}}
\newcommand{\ci}{\{c_i^1, \dots, c_i^K\}}
\newcommand{\di}{\{d_i^1, \dots, d_i^K\}}
\newcommand{\ei}{\{e_i^1, \dots, e_i^K\}}
\newcommand{\fii}{\{f_i^1, \dots, f_i^K\}}
\newcommand{\gi}{\{g_i^1, \dots, g_i^K\}}

\newcommand{\xj}{\{x^1_j, \dots, x^K_j\}}
\newcommand{\tbw}{T_{bw}}
\newcommand{\nxj}{\{\overline{x^1_j}, \dots, \overline{x^K_j}\}}
\newcommand{\nxjp}{\{\overline{x^{1\prime}_j}, \dots, \overline{x^{K\prime}_j}\}}
\newcommand{\xjp}{\{x^{1\prime}_{j}, \dots, x^{K\prime}_{j}\}}
\newcommand{\yii}[1]{\{y_{#1}^1, \dots, y_{#1}^K\}}

\newcommand{\x}[1]{\{x^1_{#1}, \dots, x^K_{#1}\}}
\newcommand{\nx}[1]{\{\overline{x^1_{#1}}, \dots, \overline{x^K_{#1}}\}}
\newcommand{\nxp}[1]{\{\overline{x^{1\prime}_{#1}}, \dots, \overline{x^{K\prime}_{#1}}\}}
\newcommand{\xp}[1]{\{x^{1\prime}_{#1}, \dots, x^{K\prime}_{#1}\}}
\newcommand{\an}[1]{\{a_{#1}^1, \dots, a_{#1}^K\}}
\newcommand{\bn}[1]{\{b_{#1}^1, \dots, b_{#1}^K\}}
\newcommand{\cn}[1]{\{c_{#1}^1, \dots, c_{#1}^K\}}
\newcommand{\gn}[1]{\{g_{#1}^1, \dots, g_{#1}^K\}}
\newcommand{\q}[1]{\{q^1_{#1}, \dots, q^K_{#1}\}}
\newcommand{\qp}[1]{\{q^{1'}_{#1}, \dots, q^{K'}_{#1}\}}
\newcommand{\pii}[1]{\{p^1_{#1}, \dots, p^K_{#1}\}}

\DeclarePairedDelimiter{\ceil}{\lceil}{\rceil}
\DeclarePairedDelimiter\floor{\lfloor}{\rfloor}
\let\epsilon=\varepsilon

\newcommand{\pred}[1]{\textsf{pred}(#1)}

{\makeatletter
 \gdef\xxxmark{%
   \expandafter\ifx\csname @mpargs\endcsname\relax 
     \expandafter\ifx\csname @captype\endcsname\relax 
       \marginpar{xxx}
     \else
       xxx 
     \fi
   \else
     xxx 
   \fi}
 \gdef\xxx{\@ifnextchar[\xxx@lab\xxx@nolab}
 \long\gdef\xxx@lab[#1]#2{{\bf [\xxxmark #2 ---{\sc #1}]}}
 \long\gdef\xxx@nolab#1{{\bf [\xxxmark #1]}}
}

\makeatletter
\newcommand{\removelatexerror}{\let\@latex@error\@gobble}
\makeatother

\newcounter{section-preserve}
\newcounter{theorem-preserve}
\newcommand{\blank}[1]{}
\newtoks\magicAppendix
\magicAppendix={}
\newtoks\magictoks
\newif\iflater
\laterfalse
\long\def\later#1{\magictoks={#1}%
  \edef\magictodo{\noexpand\magicAppendix={\the\magicAppendix \par
    \the\magictoks%
  }}
  \magictodo}
\long\def\both#1{\magictoks={#1}%
  \edef\magictodo{\noexpand\magicAppendix={\the\magicAppendix \par
    \noexpand\setcounter{theorem-preserve}{\noexpand\arabic{theorem}}%
    \noexpand\setcounter{theorem}{\arabic{theorem}}%
    \noexpand\setcounter{section-preserve}{\noexpand\arabic{section}}%
    \noexpand\setcounter{section}{\arabic{section}}%
    \noexpand\let\noexpand\oldsection=\noexpand\thesection
    \noexpand\def\noexpand\thesection{\thesection}
    \noexpand\let\noexpand\oldlabel=\noexpand\label
    \noexpand\let\noexpand\label=\noexpand\blank
    \the\magictoks%
    \noexpand\setcounter{theorem}{\noexpand\arabic{theorem-preserve}}%
    \noexpand\setcounter{section}{\noexpand\arabic{section-preserve}}%
    \noexpand\let\noexpand\thesection=\noexpand\oldsection
    \noexpand\let\noexpand\label=\noexpand\oldlabel
  }}
  \magictodo
  \the\magictoks}
\def\magicappendix{\latertrue \the\magicAppendix}

\begin{document}

\def \isnotin {\nsubseteq}

\def \eps {\varepsilon}

\bibliographystyle{alpha}

\title{\Large \textbf{Inapproximability of the Standard Pebble Game and Hard to Pebble Graphs}\footnote{A preliminary version of this paper was presented at the \emph{Algorithms and Data Structures Symposium (WADS 2017)}}}
\author[1]{Erik D. Demaine}
\author[1]{Quanquan C. Liu\footnote{This material is based
upon work supported by the National Science Foundation Graduate Research
Fellowship under Grant No. (1122374).}}
\affil[1]{Computer Science and Artificial Intelligence Lab, Massachusetts Institute of Technology,
		  Cambridge, MA, USA \authorcr
          \url{{edemaine,quanquan}@mit.edu}}
\date{}
\maketitle

\begin{abstract}
Pebble games are single-player games on DAGs involving placing and moving pebbles on nodes of the graph according to a certain set of rules. 
The goal is to pebble a set of target nodes using a minimum number of pebbles. 
In this paper, we present a possibly simpler proof of the result in~\cite{CLNV15} and strengthen the result to show that it is PSPACE-hard to determine the minimum number of pebbles to an additive $n^{1/3-\epsilon}$ term for all $\epsilon > 0$, which improves upon the currently known additive \emph{constant} hardness of approximation~\cite{CLNV15} in the standard pebble game.
We also introduce a family of explicit, constant indegree graphs with $n$ nodes where there exists a graph in the family such that using constant $k$ pebbles requires $\Omega(n^k)$ moves to pebble in both the standard and black-white pebble games. This independently answers an open question summarized in~\cite{Nor15} of whether a family of DAGs exists that meets the upper bound of $O(n^k)$ moves using constant $k$ pebbles with a different construction than that presented in~\cite{ADNV17}.
\end{abstract}

\section{Introduction}
\label{sec:intro}
Pebble games were originally introduced to study compiler operations and programming languages. For such applications, a DAG represents the computational dependency of each operation on a set of previous operations and pebbles represent register allocation. Minimizing the amount of resources allocated to perform a computation is accomplished by minimizing the number of pebbles placed on the graph~\cite{Sethi75}. The \emph{standard pebble game} (also known as the \emph{black pebble game}) is traditionally used to model such behavior. In the standard pebble game, one is given a DAG, $G = (V, E)$, with $n$ nodes and constant indegree and told to perform a set of \emph{pebbling moves} that places, removes, or slides pebbles around the nodes of $G$. 

The premise of such games is given some input modeled by \emph{source} nodes $S \subseteq V$ one should compute some set of outputs modeled as target nodes $T \subseteq V$. In terms of $G$, $S$ is typically the set of nodes without incoming edges and $T$ is typically the set of nodes without outgoing edges. The rules of the standard pebble game are as follows:

\noindent\fbox{%
    \parbox{\textwidth}{%
    \textsc{Standard Pebble Game}
    
    \textbf{Input:} Given a DAG, $G = (V, E)$. Let $\pred{v} = \{u \in V: (u, v) \in E \}$. Let $S \subseteq V$ be the set of sources of $G$ and $T \subseteq V$ be the set of targets of $G$. Let $\s = \{P_0, \dots, P_{\tau}\}$ be a valid pebbling strategy that obeys the following rules where $P_i$ is a set of nodes containing pebbles at timestep $i$ and $P_0 = \emptyset$ and $P_{\tau} = \{T\}$. Let $\p(G, \s) = \max_{i \in [\tau]}\{|P_i|\}$.\\
    
    \textbf{Rules:}
	\begin{enumerate}
		\item At most one pebble can be placed or removed from a node at a time.
		\item A pebble can be placed on any source, $s \in S$.
		\item A pebble can be removed from any vertex.
		\item A pebble can be placed on a non-source vertex, $v$, at time $i$ if and only if its direct predecessors are pebbled, $\pred{v} \in P_{i-1}$.
		\item A pebble can slide from vertex $v$ to vertex $w$ at time $i$ if and only if $(v, w) \in E$ and $\pred{w} \in P_{i - 1}$.
	\end{enumerate}
	
	\textbf{Goal:} Determine $\min_{\s}\{\p(\g, \s)\}$ using a valid strategy $\s$.
	}%
}%
\\\\

In addition to the standard pebble game, other pebble games are useful for studying computation. The \emph{red-blue pebble game} is used to study I/O complexity~\cite{HongKu81}, the \emph{reversible pebble game} is used to model reversible computation~\cite{Bennett89}, and the \emph{black-white pebble game} is used to model non-deterministic straight-line programs~\cite{CookSethi74}. Although we will be proving a result about the black-white pebble game in~\Cref{sec:hard-to-pebble}, we will defer introducing the rules of the game to the later parts of the paper since the black-white pebble game is not central to the main results of this paper. 

Much previous research has focused on proving lower and upper bounds on the \emph{pebbling space cost} (i.e. the maximum number of pebbles used at any point in time) of pebbling a given DAG under the rules of each of these games. For all of the aforementioned pebble games (except the red-blue pebble game since it relies on a different set of parameters), any DAG can be pebbled using $O(n/\log{n})$ pebbles~\cite{GT78,Hopcroft77,PTC76}. Furthermore, there exist DAGs for each of the games that require $\Omega(n/\log{n})$ pebbles~\cite{GT78,Hopcroft77,PTC76}.

It turns out that finding a strategy to optimally pebble a graph in the standard pebble game is computationally difficult even when each vertex is allowed to be pebbled only once. Specifically, finding the minimum number of black pebbles needed to pebble a DAG in the standard pebble game is PSPACE-complete~\cite{GilbertLenTar79} and finding the minimum number of black pebbles needed in the one-shot case is NP-complete~\cite{Sethi75}. In addition, finding the minimum number of pebbles in both the black-white and reversible pebble games have been recently shown to be both PSPACE-complete~\cite{CLNV15,HP10}. But the result for the black-white pebble game is proven for unbounded indegree~\cite{HP10}. A key open question in the field is whether hardness results can be obtained for constant indegree graphs for the black-white pebble game. However, whether it is possible to find good approximate solutions to the minimization problem has barely been studied. In fact, it was not known until this paper whether it is hard to find the minimum number of pebbles within even a non-constant \emph{additive} term~\cite{CLNV15}. The best known multiplicative approximation factor is the very loose $\Theta(n/\log{n})$ which is the pebbling space upper bound~\cite{Hopcroft77}, leaving much room for improvement.

Our results deal primarily with the standard pebble game, but we believe that the techniques could be extended to show hardness of approximation for other pebble games. We prove the following:

\begin{theorem}\label{thm:approx-standard}
The minimum number of pebbles needed in the standard pebble game on DAGs with maximum indegree $2$ is PSPACE-hard to approximate to within an additive $n^{1/3-\epsilon}$ for any $\epsilon > 0$.
\end{theorem}

In addition to determining the pebbling space cost, we sometimes also care about pebbling time which refers to the number of operations (placements, removals, or slides) that a strategy uses. For example, such a situation arises if we care not only about the memory used in computation but also the time of computation. 
It is previously known that there exists a family of graphs such that, given $\Theta(\frac{n}{\log{n}})$ pebbles, one is required to use $\Omega(2^{\Theta(\frac{n}{\log{n}})})$ moves to pebble any graphs with $n$ nodes in the family~\cite{LT79}. 

Less is known about the trade-offs when a small number (e.g. constant $k$) of pebbles is used until the very recent, independent result presented in~\cite{ADNV17}. It can be easily shown through a combinatorial argument that the maximum number of moves necessary using $k = O(1)$ pebbles to pebble $n$ nodes is $O(n^k)$~\cite{Nor15}. It is an open question whether it is possible to prove a time-space trade-off such that using $k = O(1)$ pebbles requires $\Omega(n^k)$ time. In this paper, we resolve this open question for both the standard pebble and the black-white pebble games using an independent construction from that presented in~\cite{ADNV17}. 

\begin{theorem}\label{thm:almost-maximum-time}
There exists a family of graphs with $n$ vertices and maximum indegree $2$ such that $\Omega(n^k)$ moves are necessary to pebble any graph with $n$ vertices in the family using constant $k$ pebbles in both the standard and black-white pebble games.
\end{theorem}

The organization of the paper is as follows. 
First, in Section~\ref{sec:definitions}, we provide the definitions and terminology we use in the remaining parts of the paper.
Then, in Section~\ref{sec:inapprox-standard}, we provide a proof for the inapproximability of the standard pebble game to an $n^{1/3 - \epsilon}$ additive factor. 

In Section~\ref{sec:hard-to-pebble}, we present our hard to pebble graph families using $k < \sqrt{n}$ pebbles and prove that the family takes $\Omega(n^k)$ moves to pebble in both the standard and black-white pebble games when $k = O(1)$.

Finally, in Section~\ref{sec:open-problems}, we discuss some open problems resulting from this paper.

\section{Definitions and Terminology}
\label{sec:definitions}
In this section, we define the terminology we use throughout the rest of the paper. All of the pebble games we consider in this paper are played on directed acyclic graphs (DAGs), and our results are given in terms of DAGs with maximum indegree $2$. We define such a DAG as $G = (V, E)$ where $|V| = n$ and $|E| = m$. 

The purpose of any pebble game is to pebble a set of targets $T \subseteq V$ using minimum number of pebbles. In all pebble games we consider, a player can always place a pebble on any source node, $S \subseteq V$. Usually, $S$ consists of all nodes with indegree $0$ and $T$ consists of all nodes with outdegree $0$.

A \emph{sequential pebbling strategy}, $\s = [P_0, \dots, P_{\tau}]$ is a series of configurations of pebbles on $G$ where each $P_i$ is a set of pebbled vertices $P_i \subseteq V$. $P_{i}$ follows from $P_{i-1}$ by the rules of the game and $P_0 = \emptyset$ and $P_{\tau} = T$. Then, by definition, $|P_i|$ is the number of pebbles used in configuration $P_i$. For a \emph{sequential} strategy, $|P_{i-1}| - 1 \leq |P_{i}| \leq |P_{i-1}| + 1$ for all $i \in [\tau] = [1, \dots, \tau]$ (i.e. at most one pebble can be placed, removed, or slid on the graph at any time). In this paper, we only consider sequential strategies.

Given any strategy $\s$ for pebbling $G$, the \emph{pebbling space cost}, $\p(G, \s)$, of $\s$ is defined as the maximum number of pebbles used by the strategy at any time:
$\p(G, \s) = \max_{i \in [\tau]}\{|P_i|\}$.

The \emph{minimum pebbling space cost} of $\sg{2}$, $\p(G)$, is  defined as the smallest space cost over the set of all valid strategies, $\mathbb{P}$, for $G$:

\begin{definition}[Minimum Pebbling Space Cost]\label{def:min-space}
\begin{align*}
\p(\sg{2}) = \min_{\s \in \mathbb{P}}\{\p(\sg{2}, \s)\}.
\end{align*}
\end{definition}

The \emph{pebbling time cost}, $\ti(\sg{2}, \s) = |\s| - 1$, of a strategy $\s$ using $s$ pebbles is the number of moves used by the strategy. The \emph{minimum pebbling time cost} of any strategy that has pebbling space cost $s$ is the minimum number of moves used by any such strategy. We know that the pebbling time cost is at least $N$ since at least $n$ moves are necessary to place a pebble on every node in an $n$ node DAG. 

\begin{definition}[Minimum Pebbling Time Cost]\label{def:min-time}
\begin{align*}
\ti(\sg{2}, s) = \min_{\s' \in \{\s \in \mathbb{P}: |\s| \leq s\}}\{\ti(\sg{2}, \s')\} \geq n.
\end{align*}
\end{definition} 

\section{Inapproximability of the Standard Pebble Game}
\label{sec:inapprox-standard}
In this section, we provide an alternative proof of the result presented in~\cite{CLNV15} that the standard pebble game is inapproximable to any constant additive factor. Then, we show that our proof technique can be used to show our main result stated in Theorem~\ref{thm:approx-standard}.  

We first introduce the PSPACE-completeness proof presented by~\cite{GilbertLenTar79} because we modify this proof to prove our main result.

\subsection{Standard Pebbling is PSPACE-complete~\cite{GilbertLenTar79}}\label{sec:previous-reduction}

The reduction is performed from the PSPACE-complete problem, quantified boolean formula (QBF). In an instance of QBF, we are given a quantified boolean formula of the form: $B = Q_1 x_1 \cdots Q_u x_u F$ where $Q_i$ is either an existential or universal quantifier, each $x_i$ is a Boolean variable, and $F$ is an unquantified boolean formula containing variable $x_i$ in CNF form with $3$ variables per clause. The decision problem is to determine whether $B$ is satisfiable for some assignment of truth values to the existential variables and for all truth assignments to the universal variables.

The reduction is done by constructing a graph $\sg{2} = (V, E)$ with one target node, $q_1$. $\sg{2}$ can be pebbled with $s$ pebbles if and only if $B$ is satisfiable. Rather than constructing gadgets to represent each variable in $F$,~\cite{GilbertLenTar79} constructs a gadget to represent each literal. For each $x_i \in F$, they create a gadget that can be set in one of three possible settings: $(x_i = True, \overline{x_i} = False)$, $(x_i = False, \overline{x_i} = True)$, or $(x_i = False, \overline{x_i} = False)$. Let $F(e_1, e_2, \dots, e_i, e_{i+1}, \dots, e_{2v -1}, e_{2v})$ represent the formula $F$ obtained by substituting $e_i$ for $x_i$ and $e_{i + 1}$ for $\overline{x_i}$ for all $i \in [v]$. Note that if $F(e_1, e_2, \dots, False, False, \dots, e_{2v-1}, e_{2v})$ is true, then, trivially, $F(e_1, e_2, \dots, True, False, \dots, e_{2v-1}, e_{2v})$ and $F(e_1, e_2, \dots, False, True, \dots, e_{2v-1}, e_{2v})$ are also both true.

The proof of their reduction from QBF relies on the following key gadgets.

\textbf{Variable Gadget:} A variable gadget is created for each variable $x_i$, $i \in [v]$ in $F$ with two paths, one path representing each literal. Each variable gadget can be in one of the following three configurations shown in Fig.~\ref{fig:glt-variable}. 

\begin{figure}[h]
\centering
\def\svgwidth{0.7\columnwidth}
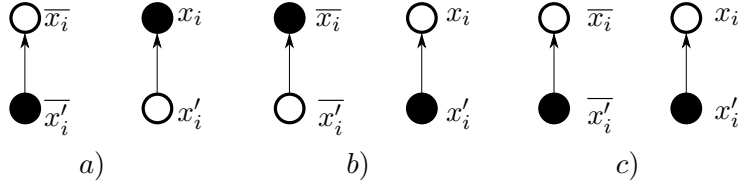
\caption{Variable gadget as used in~\cite{GilbertLenTar79}. Each variable can be in one of the $3$ configurations shown. a) True configuration. b) False configuration. c) Double false configuration. Figure recreated from~\cite{GilbertLenTar79}.}\label{fig:glt-variable}
\end{figure}

\textbf{Universal Quantifier Block:} Each universal quantifier and its associated variable is constructed in $\sg{2}$ as a universal quantifier block. There is only one way to pebble the universal quantifier block. The sequence of pebbling moves used to pebble the universal quantifier block is shown in Fig.~\ref{fig:glt-universal-block}.

\begin{figure}[h]
\centering
\def\svgwidth{0.9\columnwidth}
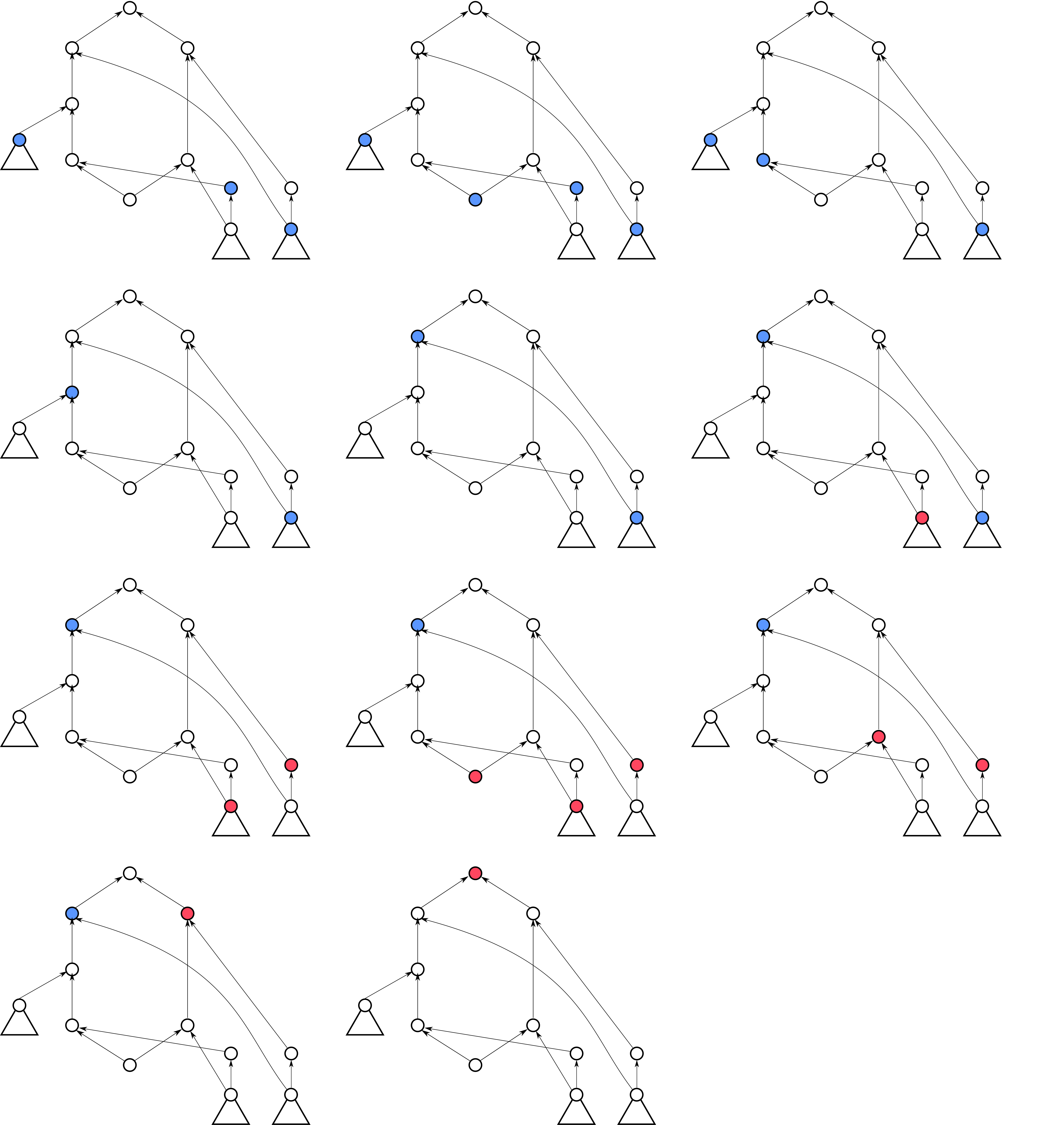
\caption{Universal quantifier block and its associated variable. There is only one way to pebble each universal quantifier block. The sequence of pebbling moves used to pebble this gadget is shown above. The change in color of the pebbles signify when all the clauses are repebbled when the universal quantifier is reset. Figure recreated from~\cite{GilbertLenTar79}.}\label{fig:glt-universal-block}
\end{figure}

\textbf{Existential Quantifier Block:} Each existential quantifier and its associated variable is constructed in $\sg{2}$ as an existential quantifier block. There are two ways to pebble the existential quantifier block depending on whether the associated variable is set to true or false. The sequence of pebbling moves used to pebble the existential quantifier block is shown in \Cref{fig:glt-existential-block}.

\begin{figure}[h]
\centering
\def\svgwidth{0.7\columnwidth}
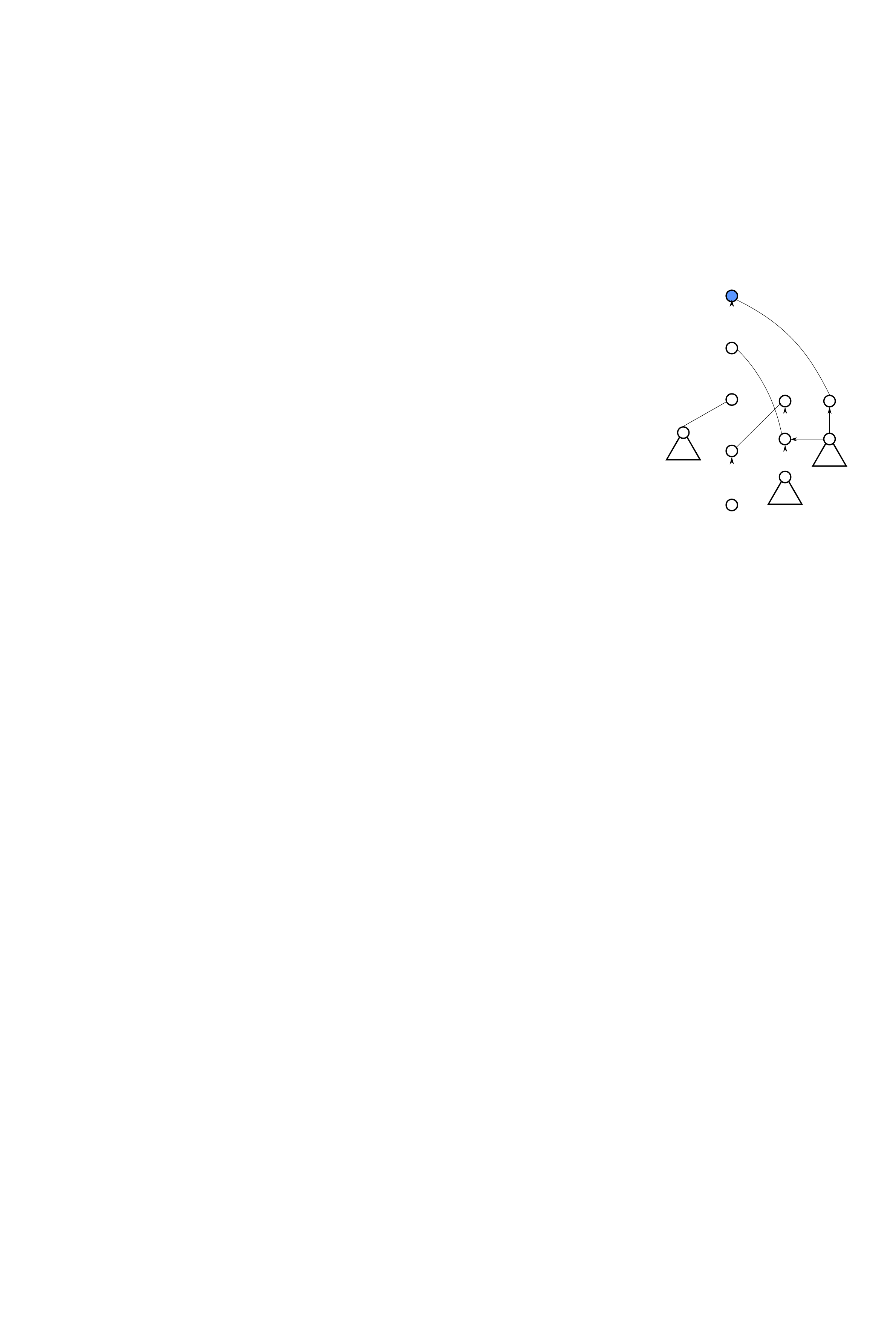
\caption{Existential quantifier block and its associated variable. There are two sequences of pebbling moves that result in $q_{i}$ being pebbled. Blue Pebbles: The sequence of pebbling moves that pebbles $q_{i}$ when the variable is set to True. Red Pebbles: The sequence of pebbling moves that pebbles $q_{i}$ when the variable is set to False. Figure recreated from~\cite{GilbertLenTar79}.}\label{fig:glt-existential-block}
\end{figure}

\textbf{Clause Gadget:} A clause gadget is created in $\sg{2}$ for each clause in $F$. The clause gadget consists of a pyramid that is connected to each literal in the clause. It was shown in ~\cite{GilbertLenTar79} that a pyramid needs as many pebbles as its height to pebble the apex. The clause gadget is shown in~\Cref{fig:glt-clause}.

\begin{figure}[h]
\centering
\def\svgwidth{0.3\columnwidth}
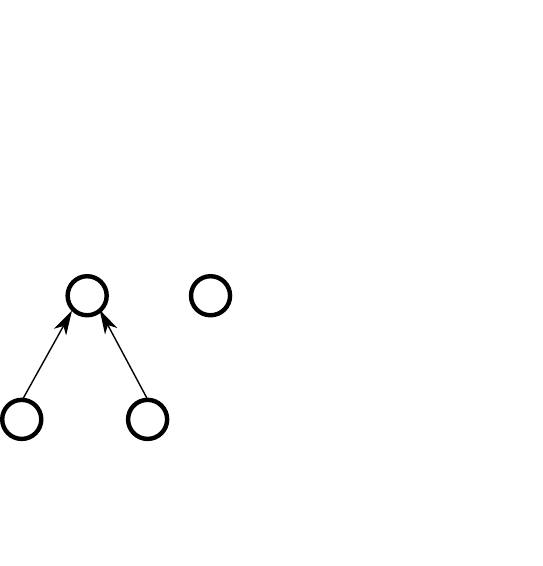
\caption{Clause gadget. The clause gadget is pebbled by first pebbling the base of the pyramid using at most $3$ extra pebbles. Figure recreated from~\cite{GilbertLenTar79}.}\label{fig:glt-clause}
\end{figure}

The entire construction using the gadgets described above can be seen in \Cref{fig:glt-entire}. 

\begin{figure}[h]
\centering
\def\svgwidth{\columnwidth}
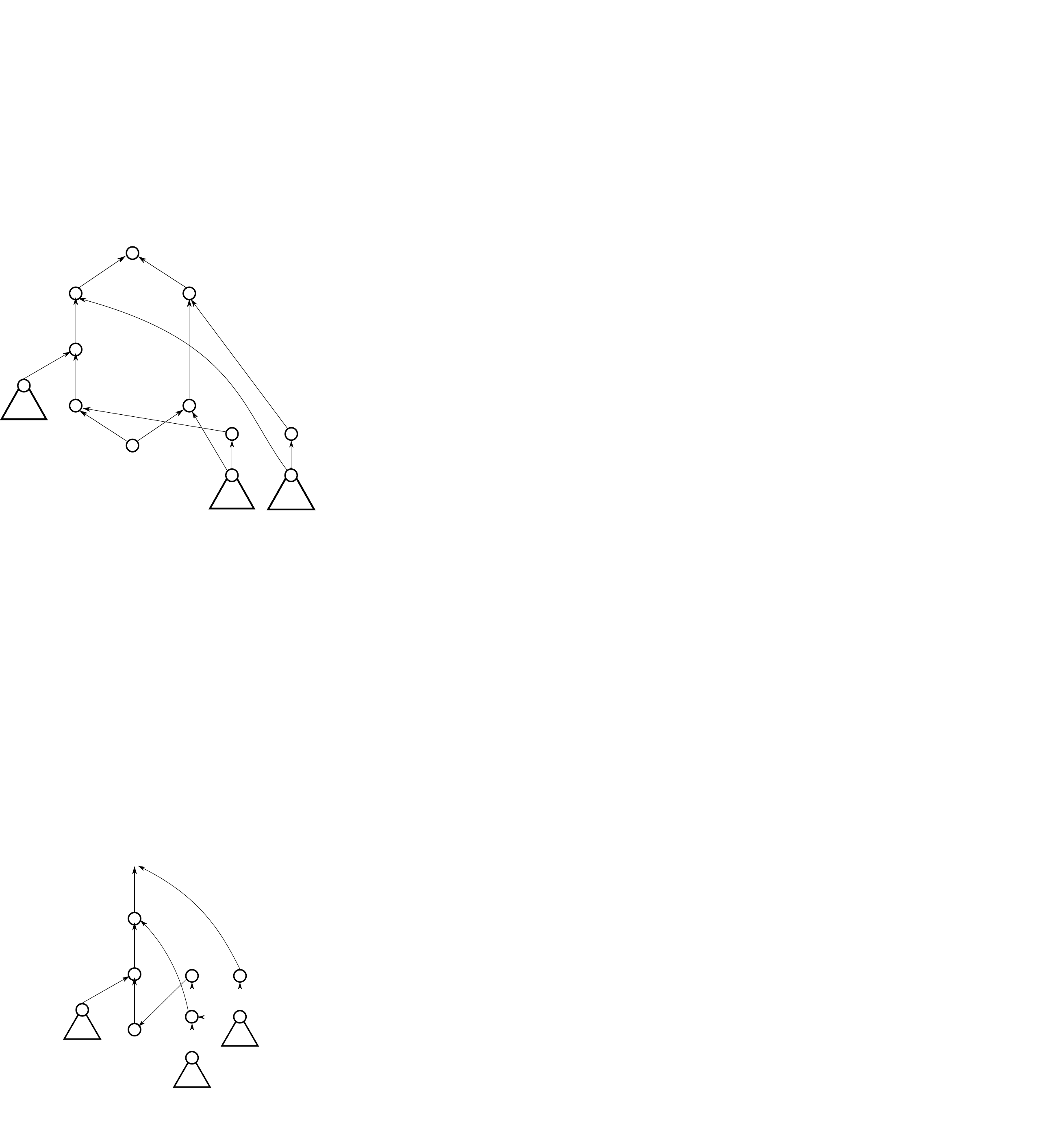
\caption{Entire construction of $\sg{2}$ for $B = \exists x_1 \forall x_2 \exists x_3 \forall x_4 \exists x_5 (\overline{x_1} \vee \overline{x_2} \vee x_3) \wedge (x_2 \vee \overline{x_3} \vee \overline{x_4}) \wedge (x_1 \vee x_4 \vee x_5) \wedge (x_3 \vee \overline{x_4} \vee \overline{x_5})$. This DAG requires $18$ pebbles to pebble, $\p(\sg{2}) = 18$, given $T = \{q_1\}$ if and only if $B$ is satisfiable. Figure recreated from~\cite{GilbertLenTar79}.}\label{fig:glt-entire}
\end{figure}

The key theorem that~\cite{GilbertLenTar79} proves that shows the PSPACE-completeness of the standard pebble game is~\Cref{thm:glt-hard}.

\begin{theorem}[Standard Pebble Game is PSPACE-hard~\cite{GilbertLenTar79}]\label{thm:glt-hard}
The quantified Boolean formula $B = Q_1x_1 \cdots Q_ux_u F$ is true if and only if vertex $q_1$ in graph $\sg{2}$ (constructed as in e.g.~\Cref{fig:glt-entire}) is pebbled with $\p(\sg{2}) = 3m + 3$ pebbles.
\end{theorem}

\textbf{Key Proof Ideas from~\cite{GilbertLenTar79}:} We describe the key proof ideas used in~\cite{GilbertLenTar79} to prove the standard pebble game PSPACE-complete that we would need to modify in order for our gap reduction described in~\Cref{sec:gap-reduction} to produce the desired inapproximability result. Using $3m + 3$ pebbles, the constructed graph is pebbled using the following series of steps. The parts that need to be modified in our proof of inapproximability are emphasized below:

\begin{enumerate}
\item Given $u$ quantifiers and their associated variables, each of the variables and pyramids in the quantifier blocks are pebbled using \textbf{\textit{$s_i$, $s_i - 1$, and $s_i - 2$ pebbles where $s_1 = 3m + 3$}}. Therefore, before the clauses are pebbled, \textbf{\textit{each of the quantifier blocks contains $3$ pebbles}}.
\item The clauses are pebbled \textbf{\textit{using the additional $3$ pebbles}} in order of construction in the graph. 
\item After all clauses have been pebbled, the remaining portions of the quantifier blocks are pebbled. All pebbles are removed from a quantifier block once that quantifier block has been pebbled. However, most quantifier blocks are pebbled more than once (some potentially $2^v$ times, see below).
\begin{itemize}
\item If the quantifier, $Q_i$, is an universal quantifier, then it can only be \textbf{\textit{pebbled in one way using $s_{i-2}$ additional pebbles}} depicted in~\Cref{fig:glt-universal-block}. The change in pebble color in~\Cref{fig:glt-universal-block} indicates when all the clauses and all quantifier blocks below it are repebbled in order to obtain a pebble on $q_{i + 1}$. This inherently means that each clause is checked for satisfiability for each setting of a universal gadget (i.e. setting $x_i$ to both False and True).
\item If the quantifier, $Q_i$, is an existential quantifier, \textbf{\textit{it can be pebbled in one of two ways using $2$ or $s_{i-2}$ additional pebbles}} depicted in~\Cref{fig:glt-existential-block} depending on whether $x_i$ is set to True or False.
\end{itemize}
\item Once the last quantifier gadget is pebbled, so is $q_1$. 
\end{enumerate}

To see more specific details of the original proof of the PSPACE-completeness of standard pebbling (including the proof of \Cref{thm:glt-hard}), please refer to~\cite{GilbertLenTar79}.

\subsection{Inapproximability to $n^{1/3-\epsilon}$ additive factor for any $\epsilon > 0$}\label{sec:gap-reduction}
We now prove our main result. For our reduction we modify all of the aforementioned gadgets in~\Cref{sec:previous-reduction}--variable gadgets, clause gadgets, and quantifier blocks.

\subsubsection{Important Subgraphs}

Before we dive into the details of our construction, we first mention two subgraphs and the properties they exhibit. 

The first graph is the \emph{pyramid graph} (shown in \Cref{fig:glt-clause} with height $4$) $\Pi_h$ with height $h$, which requires a number of pebbles that is equal to $h$ to pebble~\cite{GilbertLenTar79}. Therefore, in order to pebble the apex of such a graph, at least $h$ pebbles must be available. As in~\cite{GilbertLenTar79}, we depict such pyramid graphs by a triangle with a number indicating the height (hence number of pebbles) needed to pebble the pyramid (see~\Cref{fig:glt-entire} for the triangle symbolism).

We make use of the following definition and lemma (restated and adapted) from~\cite{GilbertLenTar79} in our proofs:

\begin{definition}[Frugal Strategy~\cite{GilbertLenTar79}]\label{def:frugal}
A pebbling strategy, $\s$, is frugal if the following are true:
\begin{enumerate}
\item Suppose vertex $v \in \sg{2}$ is pebbled for the first time at time $t'$. Then, for all times, $t > t'$, some path from $v$ to $q_1$ (the only target node) contains a pebble.
\item At all times after $v$ is pebbled for the last time, all paths from $v$ to $q_1$ contain a pebble. 
\item The number of pebble placements on any vertex $v \in \sg{2}$ where $v \neq q_1$, is bounded by the number of pebble placements on the successors of $v$. 
\end{enumerate}
\end{definition}

\begin{lemma}[Normal Pebbling Strategy (adapted from~\cite{GilbertLenTar79})]\label{lem:glt-pyramid}
If the target vertex is not inside a pyramid $\Pi_h$ and each of the vertices in the bottom level of the pyramid has at most one predecessor each, then any pebbling strategy can be transformed into a \emph{normal} pebbling strategy without increasing the number of pebbles used. A normal pebbling strategy is one that is frugal and after the first pebble is placed on any pyramid $\Pi_h$ no placements of pebbles occurs outside $\Pi_h$ until the apex of $\Pi_h$ is pebbled and all other pebbles are removed from $\Pi_h$. Furthermore, no other placement of pebbles occur on $\Pi_h$ until after the pebble on the apex of $\Pi_h$ is removed.
\end{lemma}

Note that we can tranform any pyramid that does not fit the requirements of Lemma~\ref{lem:glt-pyramid} (i.e. bottom level contains nodes with $2$ predecessors) to one that does satisfy the requirments by creating a single predecessor for each node in the bottom level and connecting the original predecessors to this single predecessor. 

The other important subgraph we use is the \emph{road graph} (\Cref{road-graph}), $R_w$ with width $w$, which requires a number of pebbles that is the width of the graph to pebble \emph{any} of the outputs~\cite{EL79,Nor15}. Therefore, we state as an immediately corollary of the provided proofs:

\begin{corollary}(Road Graph Pebbling)\label{corr:road}
To pebble $O \subseteq \{o_1, \dots, o_w\}$ of the outputs of $R_w$, with a valid strategy, $\s = [P_0, \dots, P_{\tau}]$ where $P_{\tau} = O$,  requires $w + |O| - 1$ pebbles.
\end{corollary}

\begin{figure}[h]
\centering
\includegraphics[width=0.5\textwidth]{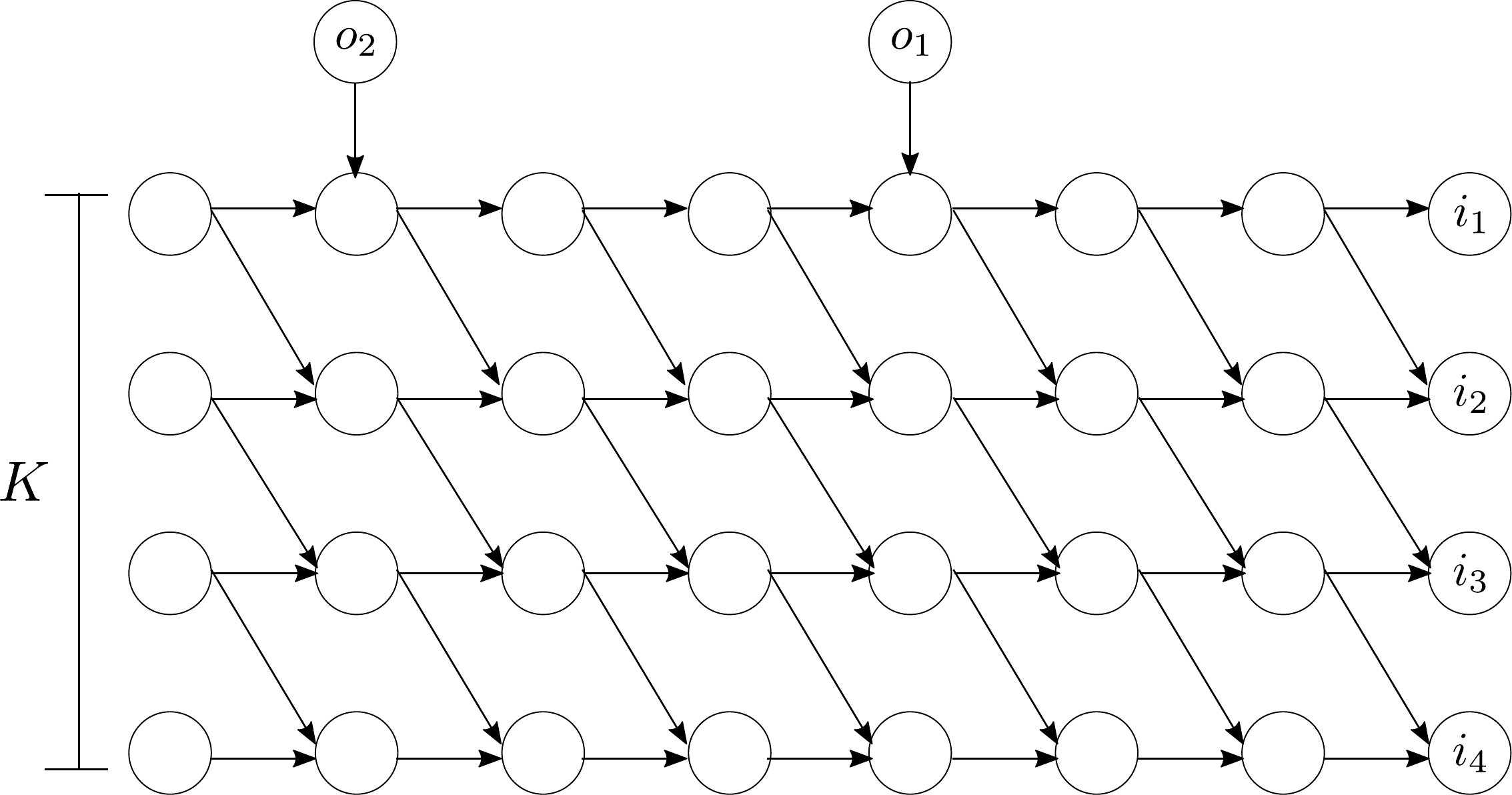}
\caption{Road graph gadget. Here, in this example, a minimum of $5$ pebbles are necessary to pebble $o_1$ and $o_2$. $4$ pebbles must be used to pebble $i_1$, $i_2$, $i_3$, and $i_4$ and one more pebble is necessary to pebble $o_1$ since the four pebbles used to pebble $i_1$, $i_2$, $i_3$ and $i_4$ must remain on the road graph in order to pebble $o_2$. $K$ is the width of this road graph gadget. In this example, $K = 4$.}\label{road-graph}
\end{figure}

We define a \emph{regular} pebbling strategy for road graphs similarly to the \emph{normal} pebbling strategy for pyramids. 

\begin{lemma}[Regular Pebbling Strategy]\label{lem:regular}
If each input, $i_j \in \{i_1, \dots, i_w\}$, to a road graph has at most $1$ predecessor, any pebbling strategy that pebbles a set of desired outputs, $O \subseteq \{o_1, \dots, o_w\}$, at the same time can be transformed into a \emph{regular} pebbling strategy without increasing the number of pebbles used. A regular pebbling strategy is one that is frugal and after the first pebble is placed on any road graph, $R_w$, no placements of pebbles occurs outside $R_w$ until the set of desired outputs of $R_w$ all contain pebbles and all other pebbles are removed from $R_w$. 
\end{lemma}

\begin{proof}
Consider any pebble strategy, $\s$, that uses $s$ pebbles. We create a regular strategy, $\s'$, that uses at most $s$ pebbles. To create $\s'$, we first delete all unnecessary pebble placements from $\s$, resulting in a frugal strategy that has no unnecessary placements and does not use more than $s$ pebbles. Suppose at time, $t_1$, a pebble is placed on a road graph, $R_w$. Let $[t_0, t_2]$ be the largest time interval containing $t_1$ such that $R_w$ is never pebble-free during $[t_0, t_2]$. Let $\mu$ be the maximum number of pebbles on $R_w$ or any of the ancestors of nodes in $R_w$ during the time interval $[t_0, t_2]$.

By our definition of the frugality of $\s'$ and the fact that the target vertices are not inside $R_w$, the only pebbles that are on $R_w$ at time $t_2$ are on the nodes in set $O$. Since at time $t_0 - 1$ no pebbles are on $R_w$, there must be a time $t_3 \in [t_0, t_2]$ at which $w + |O| - 1 \leq \mu$ pebbles are on $R_w$ by Corollary~\ref{corr:road}. Furthermore, for each input, $i_j \in \{i_1, \dots, i_w\}$, there exists a time $t_{i_j} \in [0, t_2]$ where the predecessor of $i_j$ is pebbled. We modify strategy, $\s$, in the following way to transform it into a regular strategy, $\s'$. Let $\pred{i_j}$ be the set of predecessors of $i_j$ 

\begin{enumerate}
\item Delete all pebbling placements, removals, and slides on $R_w$ in $[t_0, t_2]$.\label{delete-placements}
\item Delete all pebble placements, removals, and slides on $\pred{i_j}$ and on all ancestors of $\pred{i_j}$ for all $i_j \in \{i_1, \dots, i_w\}$ in $[t_0, t_2]$.\label{delete-pred}
\item At time $t_0$, insert a continuous sequence of moves (placements, slides, and removals deleted from~\ref{delete-pred}) that pebbles $\pred{i_j}$ (if it is not pebbled) including any pebble placements on the ancestors of $\pred{i_j}$.\label{rearrange}
\item At time $t_3$, insert a continuous sequence of moves that pebbles $O$ using $w + |O| - 1$ pebbles that pebbles all nodes in $O$ and removes all pebbles on $\pred{i_j}$. Then, insert a continuous sequence of moves that removes all pebbles on $R_w$ except for the pebbles on $O$.\label{last}
\end{enumerate}

We now prove that the above is a valid strategy that uses at most $\mu$ pebbles to pebble $R_w$ in time $[t_0, t_2]$. Steps~\ref{delete-placements}-\ref{delete-pred} do not increase the pebble count. At time $t_0 - 1$, at most $w$ $\pred{i_j}$ nodes are pebbled using strategy $\s$. Each of these pebbles is moved onto $R_w$ at some time $t_{i_j} \in [t_0, t_2]$ using $\s$. Any pebble originally on $\pred{i_j}$ during the time frame $[t_0, t_2]$ does not get removed from $R_w$ until after all nodes in $O$ are pebbled using $\s$ by construction of the road graph and since $\s$ is a frugal strategy. Therefore, steps~\ref{delete-pred}-\ref{rearrange} do not increase the pebble count using $\s'$ since after $\pred{i_j}$ is pebbled, the pebble remains on the graph. Step~\ref{rearrange} uses at most $\mu - w - |O| + 1$ additional pebbles to pebble the ancestors of $\pred{i_j}$. Step~\ref{last} uses $w+|O|-1 \leq \mu$ pebbles.
\end{proof}

We immediately obtain the following corollary from Lemmas~\ref{lem:glt-pyramid} and~\ref{lem:regular}:

\begin{corollary}\label{cor:regular-normal}
Any pebbling strategy, $\s$, can be transformed into a pebbling strategy, $\s'$, that is normal and regular if no target vertices lie inside a pyramid or road graph and each input node to either the pyramid or road graph has at most one predecessor. 
\end{corollary}

\subsubsection{Modified Graph Constructions}\label{sec:modified}
We first describe the changes we made to each of the gadgets used in the PSPACE-completeness proof presented by~\cite{GilbertLenTar79} and then prove our inapproximability result using these gadgets in~\Cref{sec:construction-proof}. Given a QBF instance, $B = Q_1 x_1 \cdots Q_u x_u F$, with $c$ clauses, we create the following gadgets:

\textbf{Variable Nodes:}
Suppose we now replace all paths in variable nodes in the proof provided by~\cite{GilbertLenTar79} (see~\Cref{fig:glt-variable}) with road graphs each of width $K$. The modified variable gadgets are shown in~\Cref{fig:modified-variable}. Each variable gadget as in the original proof by~\cite{GilbertLenTar79} has $3$ possible configurations which are also shown in~\Cref{fig:modified-variable}. 

\begin{figure}[h]
\centering
\includegraphics[width=\textwidth]{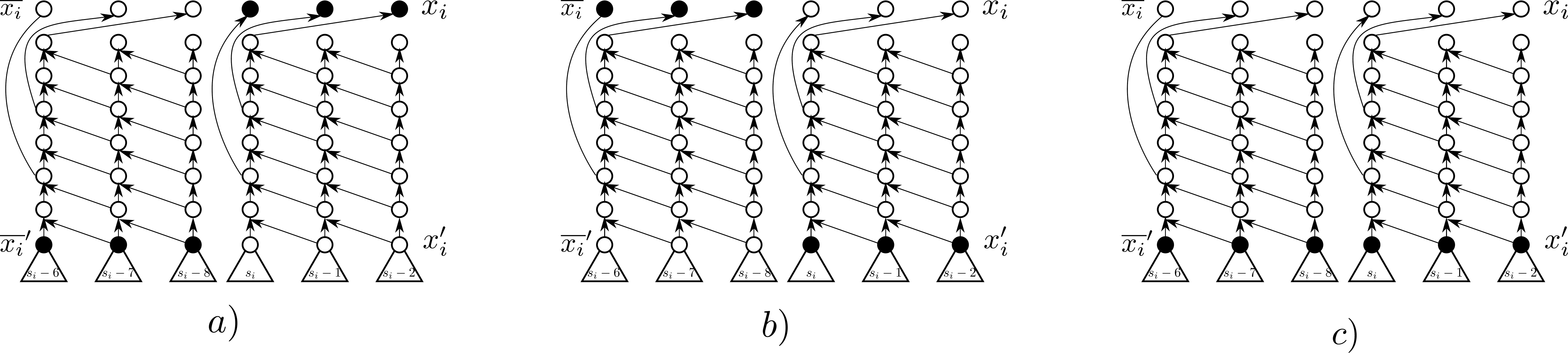}
\caption{Modified variable gadget with $3$ possible configurations using the road graph previously described. Here $K = 3$. a) $x_i$ is True. b) $\overline{x_i}$ is True. c) Double false. }\label{fig:modified-variable}
\end{figure}

\textbf{Quantifier Blocks:} Each universal and existential quantifier block is also modified to account for the new variable gadgets. See~\Cref{fig:modified-universal} and~\Cref{fig:modified-existential} which depict the new quantifier blocks that use the new variable gadgets. Note that instead of each quantifier block requiring $3$ total pebbles (as in the proof sketch described in Section~\ref{sec:previous-reduction}), each gadget requires $3K$ pebbles to remain before the clauses are pebbled. The basic idea is to expand all nodes $a_i, b_i, c_i...etc.$ into a path of length $K$ with interconnections to account for each of the $K$ copies of $x_i$ and each of $K$ copies of $\overline{x_i}$. Each $s_i = s_{i - 1} - 3K$ and $s_1 = 3Ku + 4K + 1$. 

\begin{figure}[H]
\centering
\def\svgwidth{0.9\columnwidth}
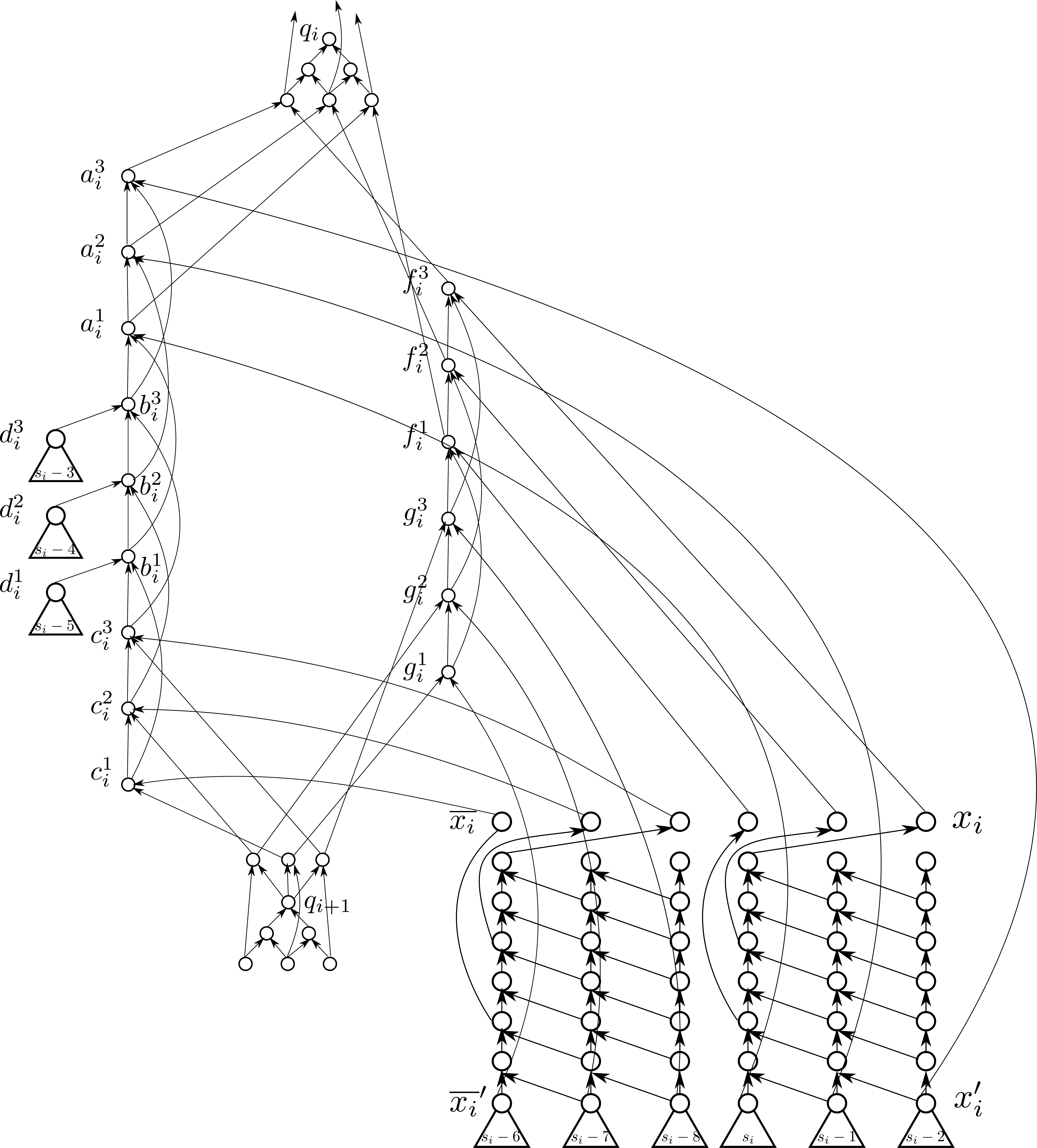
\caption{Modified universal quantifier block. Here $K = 3$.}\label{fig:modified-universal}
\end{figure}

\begin{figure}[h]
\centering
\def\svgwidth{0.8\columnwidth}
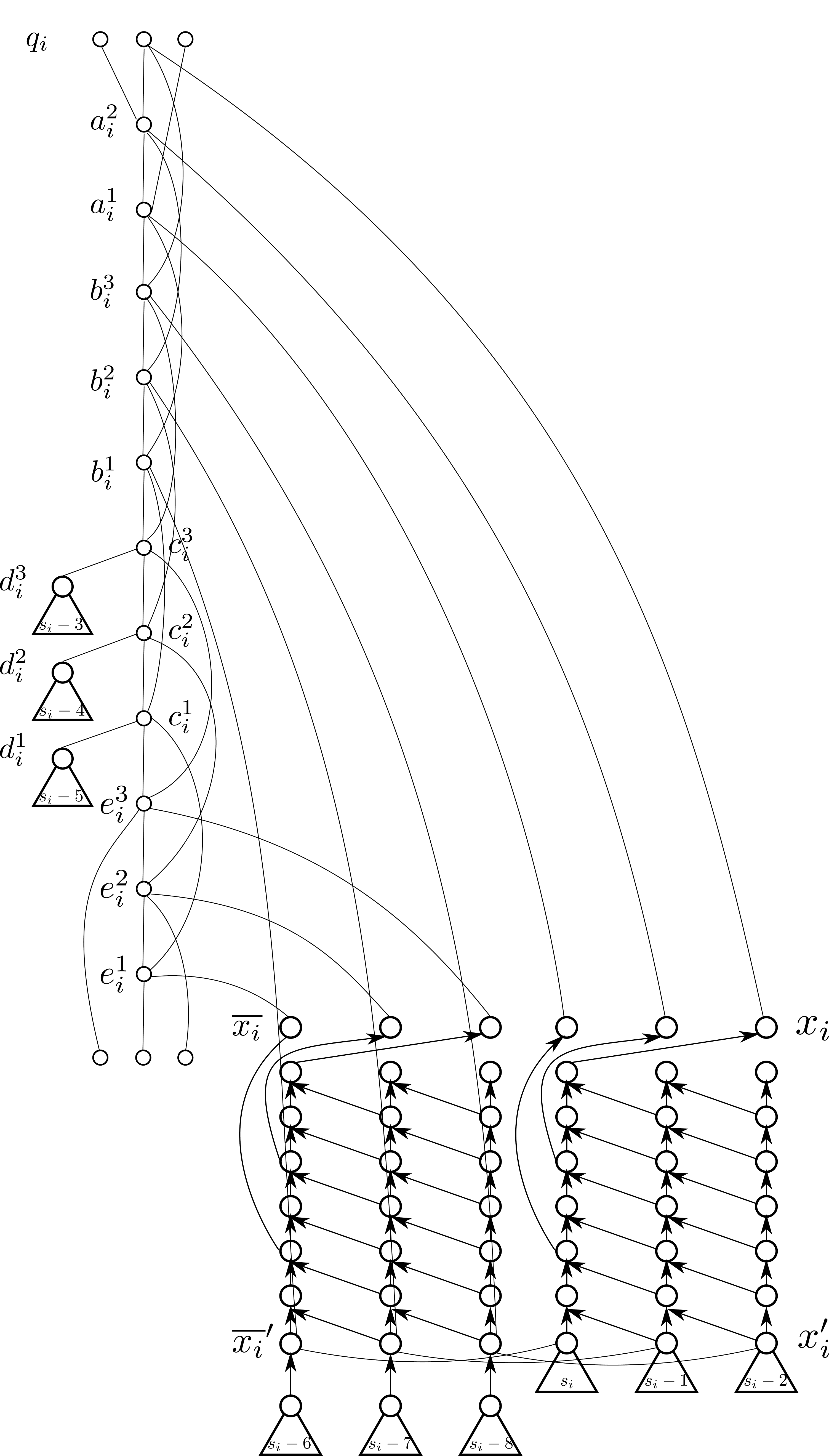
\caption{Modified existential quantifier block. Here $K = 3$.}\label{fig:modified-existential}
\end{figure}

\textbf{Clause Gadgets:} Each clause gadget is modified to be a pyramid of height $3K + 1$ where the bottom layer is connected to nodes from a combinations of different literals (as its ancestors). Therefore, for a given clause $(l_i, l_j, l_k)$, $K$ nodes are connected to $l_i$, $l_j$, $l_k$, $K$ nodes to $l_i$ and $l_k$, and $K$ nodes to $l_j$ and $l_k$. See~\Cref{fig:modified-clause} for an example of the modified clause gadget.

\begin{figure}[h]
\centering
\centering
\def\svgwidth{\columnwidth}
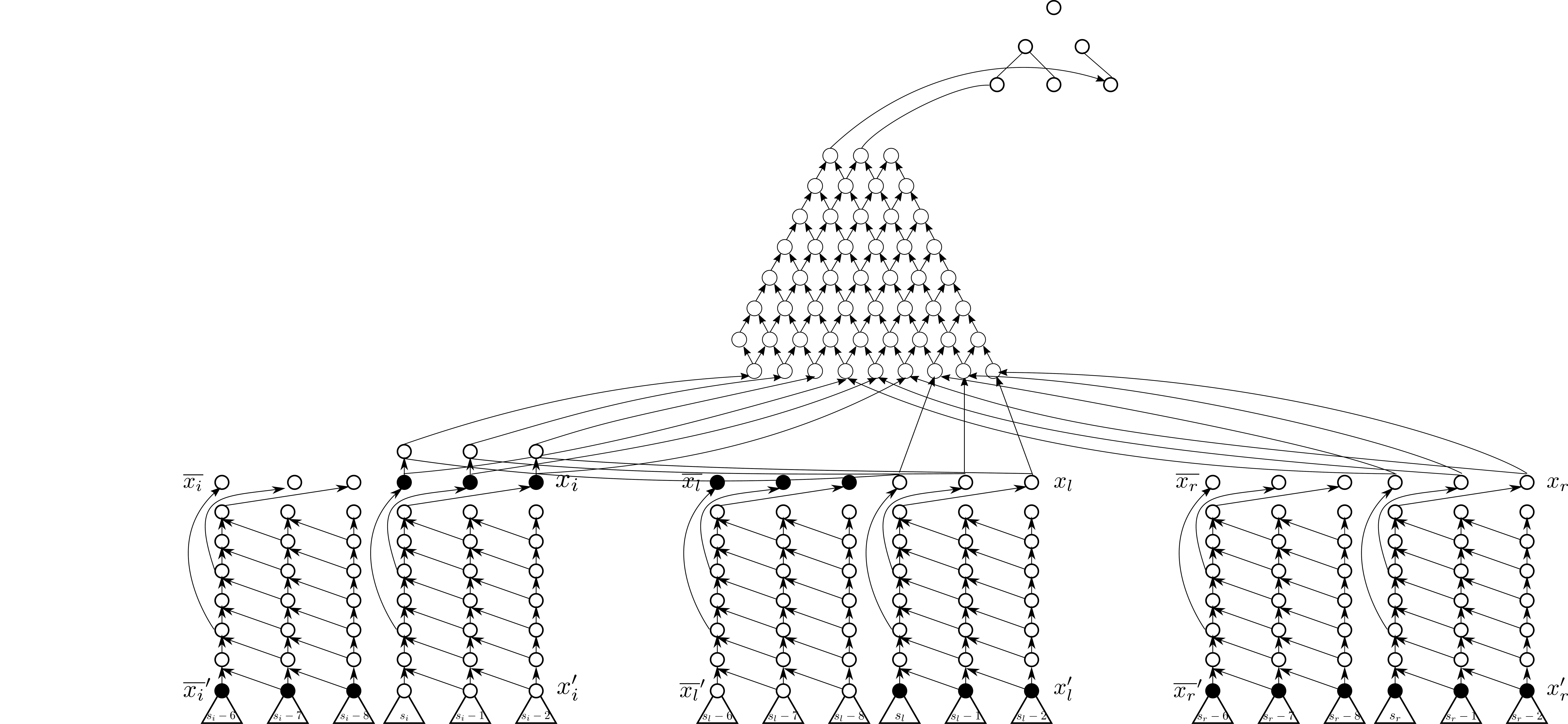
\caption{Modified clause gadget. The clause here is $(x_i, x_l, x_r)$ where $x_i = True$, $x_l = False$, and $x_r = False$. Here $K = 3$.}\label{fig:modified-clause}
\end{figure}

\begin{figure}[h]
\centering
\centering
\def\svgwidth{0.8\columnwidth}
\input{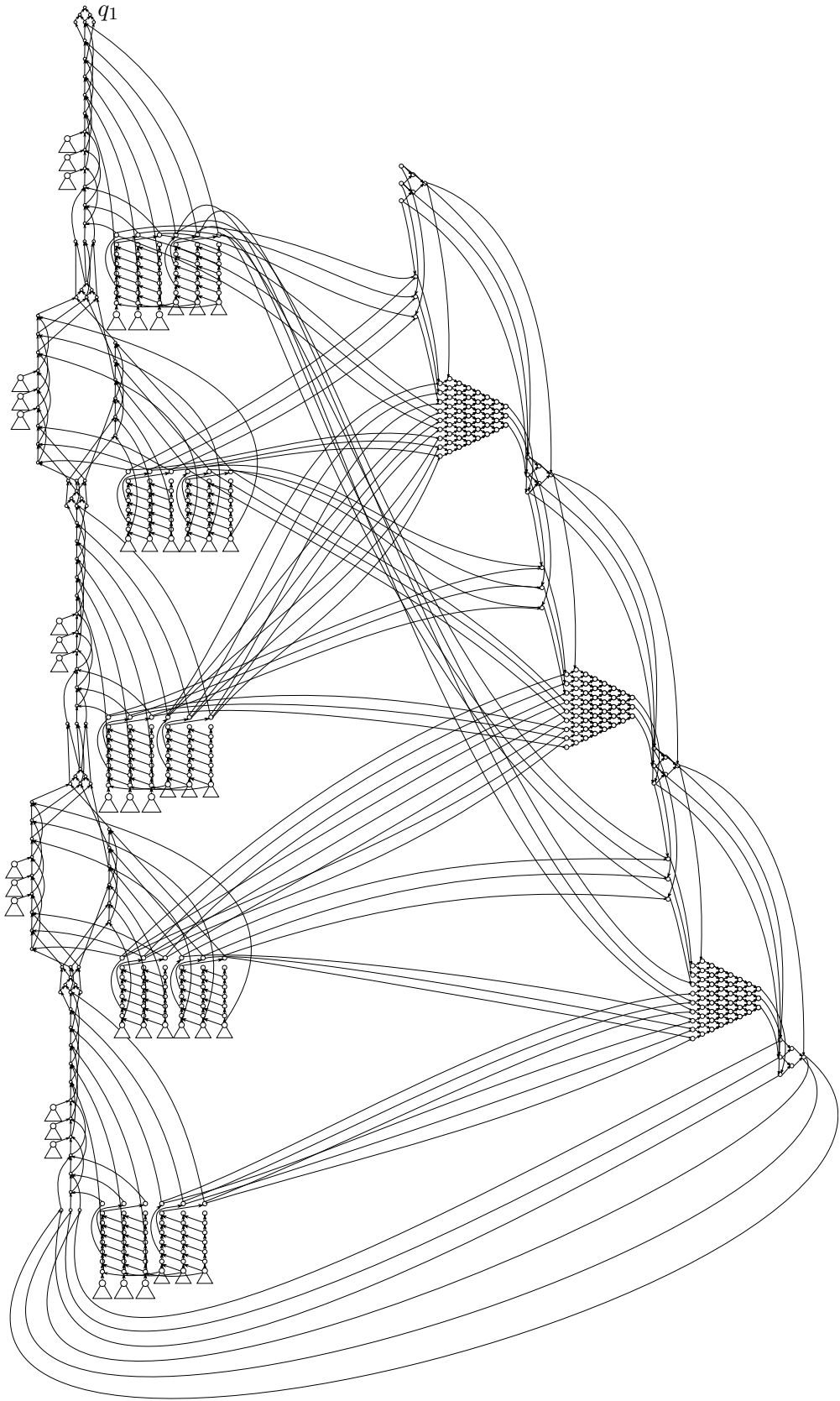}
\caption{Modified full construction. Duplicate first clause has been left out for clarity.}\label{fig:modified-full-construction}
\end{figure}

\subsubsection{Full Contruction}\label{sec:full-construction}

Using the modified gadgets described in Section~\ref{sec:modified}, we create the full construction of a graph $G$ from an instance of QBF $B$ using a similar construction to that shown in Figure~\ref{fig:glt-entire} except with the following key modifications:

\begin{enumerate}
\item $s_1 = 3Ku + 4K + 1$ and $s_i = s_{i - 1} - 3K$ for all $1 < i \leq u$. 
\item In the figures of the modified gadgets (Figs.~\ref{fig:modified-clause},~\ref{fig:modified-existential},~\ref{fig:modified-universal}, and~\ref{fig:modified-full-construction}), all vertices with indegree $3$ are replaced with pyramids of height $3$. 
\item Each $q_i$ and $p_j$ are the apex of pyramids of height $K$. Let $\q{i}$ and $\pii{j}$ be the nodes on the bottom level of each pyramid, respectively. Each $q_i$ is connected via outgoing directed edges to $K$ other nodes.  We refer to these nodes as $\qp{i}$.  There exist edges $(q_i, q^{l'}_{i})$ for all $	1 \leq l \leq K$, $(p_j, p_{j + 1})$, $(q^l_i, q^{l'}_{i})$ for all $1 \leq l \leq K$, and $(p^l_j, p^l_{j + 1})$ for all $1\leq l\leq K$. 
\item The first clause gadget in the topological sort order of clause gadgets is a duplicate of the first clause. In other words, the first two clauses gadgets are the same in the topological sort order of the clause gadgets.
\item The target node $q_1$ is the apex of a pyramid of height $K$. 
\item The clause gadget has height $3K + 1 - K = 2K + 1$ where the top $K$ nodes of the clause pyramid are connected to the corresponding $\pii{i}$. 
\end{enumerate}

\subsubsection{Proofs of the Construction}\label{sec:construction-proof}

We construct a graph $\sg{2}$ using the construction described above in~\Cref{sec:full-construction} for any given QBF instance, $B = Q_1x_1 \cdots Q_ux_uF$. In short, the proof relies on the fact that each quantifier block requires $2K$ pebbles to set the corresponding variable to true, false, or double false (i.e. the corresponding literals to true or false). An additional $K$ pebbles need to remain on each quantifier in order to be able to repebble quantifiers when checking for universal variables' satisfaction. Furthermore, a clause would consist of modified pyramids of height $3K + 1$ connected to pairs of nodes from different literals. Following the proof in~\cite{GilbertLenTar79}, the quantifier blocks are pebbled first with $3Ku$ pebbles remaining on the quantifier blocks. Then, the clauses are pebbled with $3K + 1$ pebbles. 

In the full construction, we include a duplicate copy of the first clause since the first clause can always be pebbled with $4K + 1$ pebbles regardless of whether or not it is satisfiable by the variable assignments. The variables, quantifiers blocks, and clauses are otherwise connected similarly to the construction presented in~\cite{GilbertLenTar79}.

If $B$ is satisfiable, then clauses can be pebbled with $4K + 1$ pebbles. Otherwise, $5K$ pebbles are needed to pebble one or more unsatisfied clauses in $\sg{2}$, resulting in a gap of $K-1$ pebbles between when $B$ is satisfiable and unsatisfiable. 
Thus, if given an approximation algorithm that estimates the number of pebbles needed within an additive $K-1$, we can distinguish between the case when $B$ is satisfiable (at most $3Km + 4K + 1$ pebbles are needed) and the case when $B$ is unsatisfiable (when $3Km + 5K$ pebbles are needed).

In this construction, $K$ can be any polynomial function of $u$ and $c$ where $u$ is the number of variables in $B$ and $c$ is the number of clauses (in other words, $K = u^ac^b$ for any constants $a$ and $b$). The total minimum number of pebbles necessary is $O(Ku)$ and the total number of nodes in the constructed graph is $O(K^3(u^3 + c))$.

Suppose time $t_1$ is the time when the first pebble is placed on the first clause in the topological order of the clause gadgets. Let $t_0$ be some time where $t_0 < t_1$ when $p_0$ is first pebbled. Let $t_{2}$ be the time when $q_1$ is pebbled. We first prove the following lemma.

\begin{lemma}\label{lem:remain}
At least $K$ pebbles remain on the paths from $p_0$ to $q_1$ during the time frame $[t_0, t_2]$ given a normal and regular strategy, $\s'$, that can pebble the graph using any arbitrary number of pebbles. 
\end{lemma}

\begin{proof}
For any strategy, $\s$, that pebbles the graph, by Corollary~\ref{cor:regular-normal} we show that we can transform it into a normal and regular strategy, $\s'$, where $K$ pebbles remain on the paths from $p_0$ to $q_1$ during the timeframe $[t_0, t_2]$. 

At time $t_1$, in order to pebble the first clause, $p_0$ must be completely pebbled requiring at least $K$ pebbles using strategy $\s'$. By the normality and regularity of $\s'$, $K$ pebbles remain on $\pii{i}$ for all $i \in [0, c]$. If $k \leq K$ pebbles are ever removed from $\pii{i}$, then one of three cases can occur:

\begin{enumerate}
\item If less than $K$ pebbles remain on any of the paths from $p_0$ to $p_c$ in total, then all clause gadgets $C_{i'}$ for all $i' \leq i$ must be repebbled.\label{c:one}
\item If less than $K$ pebbles remain on $\pii{i'}$ in total for all $i' \leq i$ and some $p_{i'}$ where $i' \leq i$ can be repebbled, then $K$ pebbles must remain on the top of the pyramids for clauses $C_{i'}$ where $i' \leq i$.\label{c:two}
\item If not all $\pii{i'}$ are repebbled and less than $K$ pebbles remain in total on the top of the pyramids for clauses $C_{i'}$ where $i' \leq i$, then a total of $K$ pebbles remain on $\pii{i'}$ for all $i' \leq i$.\label{c:three}
\end{enumerate} 

Case~\ref{c:one} violates the normality and regularity of $\s'$. Cases~\ref{c:two} and~\ref{c:three} both still maintain $K$ pebbles on the paths from $p_0$ to $q_1$, thus not violating the lemma.

Now we show that at the conclusion of pebbling the clauses, $K$ pebbles remain on $\pii{m} = \q{u + 1} = \qp{u+1}$. Suppose less than $K$ pebbles remain on $\pii{m}$, then $Q_{u}$ trivially cannot be pebbled and one of the three cases above must be true. 

We now argue that $K$ pebbles remain on either $\q{i + 1}$, $\ci$, $\bi$, $\ai$, $\gi$, or $\fii$ during the pebbling of a universal quantifier block or on $\q{i + 1}$, $\ai$, $\bi$, $\ci$, or $\ei$ during the pebbling of an existential quantifier block. 

When pebbling the first quantifier block, we consider two cases:

\begin{enumerate}
\item Universal quantifier block: $K$ pebbles must be moved from  $\qp{u + 1}$ to $\ci$.  When pebbling a $False$ assignment to the variable, the $K$ pebbles must remain on $\ci$ until the last of the nodes in the topological sort order of $\ci$ is pebbled since $\bi$ cannot be pebbled until the last of the $\ci$ is pebbled and one of the predecessors of $b_i^1$ is $c_i^1$. By the same argument, $K$ pebbles are moved onto $\bi$ then $\ai$ and remain on either $\ai$ or $\bi$ until $\q{u}$ is pebbled. When pebbling a $True$ assignment to the variable, we follow the same logic that $K$ pebbles must be moved onto $\gi$ and then $\fii$ until $K$ pebbles are used to pebble $\q{u}$. 

\item Existential quantifier block: $K$ pebbles must be moved from $\qp{u + 1}$ to $\ei$. No pebbles can be removed from $\ei$ until the final element in the topological sort order of $\ei$ is pebbled since to pebble the first element of $\ci$, we need to pebble the last element of $\ei$. Then, $K$ pebbles are transferred from $\ei$ to $\ci$ to $\bi$, to $\ai$, and finally to $\q{u}$. 
\end{enumerate}

The rest of the proof for $\q{i}$ where $1 \leq i < u$ follows easily from induction by using the base case provided above. 
\end{proof}

We then prove that the number of pebbles needed to pebble each quantifier block is $3K$ and $3K$ pebbles remain on the quantifier blocks throughout the pebbling of the clauses. 

\begin{lemma}\label{lem:quant}
Let $N_i$ be the configuration such that some number of pebbles are on the first $i - 1$ quantifier blocks and the $i$-th quantifier block (i.e. $\q{i+1}$ contains $K$ pebbles and $\q{i}$ does not yet contain $K$ pebbles) is being pebbled. Therefore, $N_{u + 1}$ is the configuration when some number of pebbles are on all $u$ quantifier blocks and the first clause gadget is being pebbled. There does not exist a normal and regular strategy, $\s$, that uses less than $3Ku + 5K$ pebbles that can pebble our reduction construction, $\sg{2}$, such that $N_{i}$ contains less than $s - s_i$ pebbles on the first $i - 1$ quantifier blocks when the $i$-th quantifier block or when the first clause is being pebbled.
\end{lemma}

\begin{proof}
Let $Q_i$ be the $i$-th quantifier block. If $N_{u + 1}$ contains less than $s - s_{u+1} = 3Ku$ pebbles, then there exists a quantifier block, $Q_j$, that contains less than $3K$ pebbles when the clauses are being pebbled. Let $Q_j$ be the first quantifier block that is missing at least one pebble (i.e. $N_j$ contains $s - s_j$ pebbles). If less than $3K$ pebbles are on the $Q_j$ block then two possible scenarios could occur. Either:

\begin{enumerate}
\item Some literal configuration contains less than $K$ pebbles. 
\item $d_j^l$ is not pebbled for some $1\leq l \leq K$. 
\end{enumerate}

Let $t_c$ be the time when the first clause is pebbled. Suppose on the contrary that less than $3K$ pebbles were placed on each quantifier block at time $t_{c} - 1$. Let the unpebbled vertex be $v$ and let $v$ be part of quantifier block $Q_j$. If $v$ is part of a literal in a clause, then $v$ must be pebbled in order to pebble the clauses. Let time $t'$ be the time when a clause containing $v$ is pebbled. Then, $v$ must be pebbled at $t'$. To pebble $v$ at time $t'$,  at least $s' \geq s_j - 3K + 1$ pebbles must be removed from the graph to pebble the literal. 

Without loss of generality, we assume that the apex of the pyramid with height $s_i - 3K + 1$ is missing a pebble at vertex $v$. Note that our argument applies for any of the pyramids in the gadget that is missing a pebble. For any other pyramid that is missing a pebble in $Q_{j}$, we need only consider whether pebbles remain on $Q_j$ itself.

Since our strategy $\s$ uses less than $3Ku + 5K$ pebbles in total to pebble $G$, there are two ways to obtain the necessary $s'$ pebbles:

\begin{enumerate}
\item Remove $s'$ pebbles from clauses and quantifier blocks $Q_{j'}$ where $j' > j$. Then, we must remove all pebbles except for at most $K - 1$ pebbles that can remain on the clause gadgets or the quantifier blocks $Q_{j'}$. Thus we know that at most $K-1$ pebbles total can be on the paths from $p_0$ to $p_c$. By Lemma~\ref{lem:remain}, this contradicts the normality and regularity of $\s$. \label{c:r-one}
\item Remove pebbles from $Q_{j'}$ where $j' < j$. Then, our argument for Case~\ref{c:r-one} applies to $Q_{j'}$ and so on until no quantifier block with lower order number contains less than $3K$ pebbles. 
\end{enumerate}

Thus if $N_{u+1} < 3Ku$, then the normality and regularity of $\s$ is violated. 

If $d_j^i$ is not pebbled for some $1 \leq i \leq K$, then, when $d_j^i$ must be pebbled, at least $s_i - 2K + 1$ pebbles must be removed from the graph in order to pebble $d_j^i$. Lemma~\ref{lem:remain} is again violated as at most $K-1$ pebbles total can remain on the paths from $p_0$ to $q_1$. 
\end{proof}

Next we prove that provided $3Km$ pebbles stay on the quantifier blocks, each unsatisfied clause requires $5K$ pebbles.

\begin{lemma}\label{lem:clause}
Given a clause gadget, $C_i$, where $i > 1$ (since the first clause is a duplicate), its corresponding clause, $c_i$ is true if and only if $C_i$ can be pebbled with $4K + 1$ pebbles (including the $K$ pebbles on $p_{i-1}$) and no pebbles are added, removed, or slid on the literals attached to the clause gadget. Furthermore, if all literals in $C_i$ are set in the false configuration, then at least $5K$ pebbles are necessary to pebble the clause (including the $K$ pebbles on $p_{i-1}$). 
\end{lemma}

\begin{proof}
We first prove that if $c_i$ is true, then $C_i$ can be pebbled using $4K + 1$ pebbles. Given any valid strategy $\s$, we first note that we can transform $\s$ into a regular strategy where the pyramid in~\Cref{fig:modified-clause} with apex $p_j$ can be pebbled using $4K + 1$ pebbles. Because $c_i$ is true, at least one of the literals in $c_i$ must be true. Therefore, one of the literals connected to the bottom of the pyramid in $C_i$ must be in the true configuration. $4K$ pebbles can then be used to pebble the other two literals, $\pii{j-1}$, and $3K + 1$ pebbles can then be used to pebble the clause pyramid. One can check a small number of cases to see that this is true. For example, in Fig.~\ref{fig:modified-clause}, we have the following cases:

\begin{enumerate}
\item Suppose that $x_i = True$, then using $3K$ pebbles, we pebble $\x{l}$ and $\x{r}$, leaving $2K$ pebbles. We pebble the nodes in the bottom layer of the pyramid that have $\x{l}$ and $\x{r}$ as predecessors with the $K$ remaining pebbles. Then, we move the $K$ pebbles on $\x{l}$ to $\yii{i}$ (assuming by Lemma~\ref{lem:remain} that $\pii{j-1}$ contains $K$ pebbles) and to the bottom level of the clause pyramid. Finally, we move the $K$ pebbles from $\x{r}$ to the bottom level of the clause pyramid. 
\item Suppose that $x_l = True$, then using $3K$ pebbles, we pebble $\x{i}$ and $\x{r}$. Use remaining $K$ pebbles to pebble nodes on bottom layer of clause pyramid that have $\x{i}$ and $\x{r}$ as predecessors. Move $K$ pebbles from $\x{i}$ to $\yii{i}$ and then to bottom layer of clause pyramid. Move $K$ pebbles from $\x{r}$ to bottom layer. 
\item Suppose that $x_r = True$, then using $3K$ pebbles, we pebble $\x{i}$ and $\x{l}$. Use remaining $K$ pebbles to pebble $\yii{i}$ and then bottom layer of clause pyramid. Move $K$ pebbles from $\x{i}$ to bottom layer. Move $K$ pebbles from $\x{l}$ to bottom layer. 
\end{enumerate}

Now we prove the more difficult direction that if $C_i$ can be pebbled using $4K + 1$ pebbles, then $c_i$ is true. Each variable gadget requires $K$ pebbles to pebble each output. There exists a time $t_1$ when $3K + 1$ pebbles are on the clause pyramid. By regularity of the pyramid, a total of $3K + 1$ pebbles must be on the predecessors of the bottom level of the pyramid before any pebbles are placed on the pyramid. However, at most $2K + 2$ pebbles can be placed on the pyramid if none of the literals are set in the true configuration. $K$ pebbles must be on $p_{i - 1}$ or the frugality of the strategy is violated as proven in Lemma~\ref{lem:remain}. The only entry point into the pyramid is through the literals. However, every output requires $K$ pebbles to pebble. For every pebble that enters the pyramid, at least $K$ pebbles must be free to pebble the road graph. However, if none of the literals are true, then to place the $p$-th pebble where $2K + 3 \leq p \leq 3K$ on the pyramid requires one additional pebble each which we remove from other parts of the graph, a contradiction to our assumption. 

Following the argument presented above, to pebble each output of a literal requires $K$ pebbles. Therefore, to move the $3K + 1$ pebbles onto the clause pyramid, we need $K - 1$ additional pebbles, resulting in $5K$ pebbles. 
\end{proof}

Given the previous proofs, we now prove the following key lemmas: 

\begin{lemma}\label{lem:upper-bound-true}
Given $\sg{2}$ which is constructed from the provided QBF instance, $B = Q_1x_1 \cdots Q_u x_u F$, using our modified reduction in~\Cref{sec:modified}, $B$ is satisfiable if and only if $\p(\sg{2}) \leq 3Ku + 4K + 1$. 
\end{lemma}

\begin{proof}
We first prove that if $B$ is satisfiable, then the graph can be pebbled with $3Ku + 4K + 1$ pebbles. We prove this via induction, similar to the proof given in~\cite{GilbertLenTar79}. Let $s = 3Ku + 4K + 1$ and $s_i$ be defined as in~\Cref{sec:full-construction}. For $1 \leq i \leq u + 1$, we define $N_i$ to be the set of configurations fixing truth values to the first $i - 1$ variable nodes. An arrangement of exactly $s - s_i$ pebbles on $\sg{2}$ is in $N_i$ if and only if, for $1 \leq j < i$, the following two conditions hold:

\begin{enumerate}
\item If $Q_j = \forall$, then exactly $3K$ pebbles are on the $j$-th quantifier block, on one of the following three sets of vertices:
\begin{enumerate}
\item $\{d_j^1, \dots, d_j^K\} \cup \xj \cup \nxjp$, indicating $x_j$ is $True$;
\item $\{d_j^1, \dots, d_j^K\} \cup \nxj \cup \xjp$, indicating $x_j$ is $False$;
\item $\{d_j^1, \dots, d_j^K\} \cup \xjp \cup \nxjp$, indicating a double $False$.
\end{enumerate}

\item If $Q_j = \exists$, then exactly $3K$ pebbles are on the $j$-th quantifier block, on one of the following three sets of vertices:

\begin{enumerate}
\item $\{d_j^1, \dots, d_j^K\} \cup \xj \cup \nxjp$ indicating $x_j$ is $True$;
\item $\{d_j^1, \dots, d_j^K\} \cup \nxj \cup \xjp$, indicating $x_j$ is $False$;
\item $\{d_j^1, \dots, d_j^K\} \cup \xjp \cup \nxjp$, indicating a double $False$.
\end{enumerate}
\end{enumerate}

By our definition, $N_1$ contains no pebbles on the graph, and $N_{u + 1}$ contains all configurations in which a truth assignment has been made to each literal and $4K + 1$ pebbles remain to test whether the assignment makes $F$ true.

We now prove the following claim which subsequently also proves that if $B$ is satisfiable then $\p(G) \leq 3Ku + 4K + 1$.

\begin{claim}\label{clm:forward}
Let $1 \leq i \leq u+1$. Suppose the graph is initially in a configuration $N_i$. For $1 \leq j < i$, let $e_{2j - 1}$ be the truth assignment defined for $x_j$ by that configuration, and let $e_{2j}$ be the truth assignment defined for $\overline{x_j}$. If $Q_i x_i \cdots Q_u x_u F(e_1, e_2, \cdots, e_{2i - 3}, e_{2i - 2})$ is true, then vertex $q_i$ can be pebbled with additional pebbles without moving any of the $s - s_i$ pebbles initially on the graph.
\end{claim}

\begin{proof}
We prove by induction on $i$ from $u+1$ to $1$.

Let $i = u+1$ and suppose that the assignment defined by the $N_i$ configuration makes $F$ true. We must show that any vertex $\qp{u + 1} = \pii{c}$ can be pebbled with $s_{u + 1} = 4K + 1$ pebbles without moving any of the pebbles of the $N_{u + 1}$ configuration. We showed in Lemma~\ref{lem:clause} that this is the case. 

Now suppose that the lemma holds for $i + 1$ so that the assignment defined by the $N_i$ configuration makes the substituted formula $Q_i x_i \cdots Q_u x_u F(e_1, e_2, \cdots, e_{2i - 3}, e_{2i - 2})$ true. 

To prove that the lemma holds for $i$, there are two cases we have to consider:

\begin{enumerate}
\item Suppose $Q_i = \forall$. Then, 
\begin{align*}
Q_{i + 1}x_{i+1} \cdots Q_u x_u F(e_1, \cdots, e_{2i - 2}, True, False) \text{ and }\\
Q_{i+1}x_{i+1} \cdots Q_u x_u F(e_1, \cdots, e_{2i - 2}, False, True)
\end{align*}

are both true. 

Vertices $q_i$, $\q{i}$, and $\qp{i}$ can be pebbled with $s_i$ pebbles as follows. First, use all $s_i$ pebbles to pebble $\xp{i}$, leaving $K$ pebbles, one on the apex of each of the pyramids. Then, use the remaining $s_i - K$ pebbles to pebble $\{d_i^1, \dots, d_i^K\}$ leaving $K$ pebbles, one on each of the $d_i$'s. Finally, with $s_i - 2K$ remaining pebbles, pebble $\nxp{i}$, then move the $K$ pebbles on $\nxp{i}$ to $\nx{i}$. The current configuration is in $N_{i + 1}$ representing the variable $x_i = False$. Applying the induction hypothesis, pebble $q_{i + 1}$, $\q{i}$, and $\qp{i}$ with the remaining $s_{i + 1} = s_{i} - 3K$ pebbles. Move the pebbles on $\qp{i + 1}$ to $\{c_i^1, \dots, c_i^K\}$, $\{b_i^1, \dots, b_i^K\}$, and $\{a_i^1, \dots, a_i^K\}$. Move the pebbles on $\xp{i}$ to $\x{i}$. Leaving pebbles on $\{a_i^1, \dots, a_i^K\}$ and $\x{i}$, pick up the rest of the pebbles and use the $s_i - 2K$ free pebbles to pebble $\nxp{i}$, leaving $K$ pebbles there. The current configuration is in $N_{i + 1}$, representing the variable $x_{i} = True$. Applying the induction hypothesis, pebble $\qp{i + 1}$ again. Finish by moving the $K$ pebbles on $\qp{i + 1}$ to $\{g_i^1, \dots, g_i^K\}$, $\{f_i^1, \dots, f_i^K\}$, and $\q{i}$. 

If $Q_{i + 1}x_{i + 1} \cdots Q_u x_u F(e_1, \cdots, e_{2j - 2}, False, False)$ is true, there is a way to pebble $\q{i}$ that only pebbles $\qp{i + 1}$ once. First pebble $\xp{i}$, $\{d_i^1, \dots, d_i^K\}$, and $\nxp{i}$, which gives a configuration in $N_{i + 1}$ representing the literals $x_i$ and $\overline{x_i}$ are both false. Applying the induction hypothesis, pebble $\qp{i + 1}$. There are now $s_i - 4K \geq 2K$ free pebbles. Place $K$ on $\nx{i}$ and move them to $\cn{i}$, $\bn{i}$, and $\an{i}$. Move the pebbles on $\nxp{i}$ to $\gn{i}$ and finish by moving the pebbles on $\xp{i}$ to $\x{i}$, $\fii$, and $\q{i}$. 
\item Suppose $Q_i = \exists$. Then either

\begin{align*}
Q_{i + 1}x_{i + 1} \cdots Q_u x_u F(e_1, \cdots, e_{2i - 2}, True, False) \text{ or }\\
Q_{i + 1}x_{i + 1} \cdots Q_u x_u F(e_1, \cdots, e_{2i - 2}, False, True) 
\end{align*} 

is true. 

Suppose that the former is the case. Vertices $\q{i}$, $q_{i}$, and $\qp{i}$ can be pebbled with $s_i$ pebbles as follows. First pebble $\xp{i}$ leaving pebbles there. Then, pebble $\di$ and $\fii$, leaving pebbles there. Move the pebbles on $\fii$ to $\nxp{i}$ and move the pebbles on $\xp{i}$ to $\x{i}$. The current configuration is in $N_{i + 1}$, representing variable $x_{i} = True$. Applying the induction hypothesis, pebble $\qp{i + 1}$ with the remaining $s_{i + 1} = s_i - 3K$ pebbles. There are now $s_i - 4K \geq 2K$ free pebbles. Place $K$ pebbles on $\nx{i}$ and finish by moving the $K$ pebbles on $\nx{i}$ to $\ei$, $\ci$, $\bi$, $\ai$, and $\q{i}$. 

Alternatively, suppose that $Q_{i + 1}x_{i + 1} \cdots Q_u x_u F(e_1, \cdots, e_{2i - 2}, False, True)$ is true. To pebble $\q{i}$ with $s_i$ pebbles, begin by pebbling $\xp{i}$, $\di$, and $\fii$ in turn, leaving pebbles there. Move the pebbles on $\fii$ to $\nxp{i}$  and $\nx{i}$, which gives a configuration in $N_{i + 1}$ representing variable $x_i = False$. Applying the induction hypothesis, pebble $\qp{i + 1}$. Move the pebbles on $\qp{i + 1}$ to $\ei$ and $\ci$. Pick up all the pebbles except for those on $\ci$ and $\xp{i}$ and use the $s_i - 2K$ free pebbles to pebble $\fii$. Move the pebbles on $\fii$ to $\nxp{i}$, then $\bi$, and finish by moving the pebbles on $\xp{i}$ to $\x{i}$, then $\ai$, and finally to $\q{i}$. 
\end{enumerate}
\end{proof}

When $i = 1$, the proof of of the above Claim~\ref{clm:forward} proves that if $B$ is satisfiable then $\p(G) \leq 3Ku + 4K + 1$. 

Now we prove that if $\p(\sg{2}) \leq 3Km + 4K + 1$, then $B$ is satisfiable. Note that in the subsequent proofs, we assume that we only remove pebbles from quantifiers blocks with higher order number since we proved in Lemma~\ref{lem:remain} that removing pebbles from quantifier blocks with lower order number violates the normality and regularity of a pebbling strategy. We first prove the following claim which subsequently proves this.

\begin{claim}
Let $1 \leq i \leq n + 1$. Suppose the graph is initially in a configuration in $N_i$. For $1 \leq j < i$, let $e_{2j - 1}$ be the truth assignment defined for $x_i$ by that configuration, and let $e_{2j}$ be the truth assignment for $\overline{x_j}$. If vertex $q_i$ can be pebbled with $s_i$ additional pebbles without moving any of the $s - s_i$ pebbles initially on the graph, then $Q_i x_i \cdots Q_n x_n(e_1, e_2, \cdots, e_{2i-3}, e_{2i-2})$ is true.
\end{claim}

\begin{proof}
We assume that any strategy $\s$ that can pebble $G$ using $3Ku + 4K + 1$ pebbles is transformed into a normal and regular strategy $\s'$ that pebbles $G$ using $3Ku + 4K + 1$ pebbles. Again, we prove by induction on $i$ from $u + 1$ to $1$. In the base case, let $i = u+1$, then by Lemma~\ref{lem:quant} and Lemma~\ref{lem:clause}, each clause gadget contains at least one literal gadget in the $True$ configuration. 

Suppose by induction that the lemma holds for $i+1$, we now prove there is a strategy which pebbles $\q{i}$ with $s_i$ pebbles without moving any pebbles in the $N_i$ configuration. We can assume that such a strategy is normal and regular by Corollary~\ref{cor:regular-normal}. We now consider $Q_i$, the $i$-th quantifier gadget.

\begin{enumerate}
\item Suppose $Q_i = \forall$. Suppose that $t_0$ is a time when $s_i$ pebbles appear on the $s_i$-pyramid. After $t_0$, each of $\xp{i}$ is only pebbled once before $\q{i}$ are pebbled. Furthermore, by frugality, $\ai$, $\bi$, $\ci$, $\di$, $\fii$, and $\gi$ are each pebbled only once after $t_0$ until $\q{i}$ are pebbled. Let $t_1$ be the time when $x^{K'}_i$ is pebbled. Since our strategy is a normal and regular strategy, by Lemma~\ref{lem:quant} from $t_1$ until when $\qp{i + 1}$ is pebbled, $K$ pebbles remain on $\xp{i}$, $\x{i}$, or $\fii$. From $t_1$ until $\ai$ are pebbled, $K$ pebbles are on $\xp{i}$. 

To pebble $a^1_i$ requires pebbling $\di$. This requires removing all pebbles from the block except the $K$ pebbles on $\xp{i}$. By the normality of pebbling strategies, $\di$ are pebbled before everything else in the gadget besides $\xp{i}$ and $K$ pebbles remain on $\di$ until $\bi$ are pebbled. To pebble $\bi$ requires pebbling $\ci$ which subsequently requires pebbling $\nxp{i}$. To pebble $\nxp{i}$ requires removing all pebbles on the gadget except those on $\xp{i}$ and $\di$. Therefore, $\nxp{i}$ are pebbled immediately after $\di$ and $K$ pebbles remain on $\nxp{i}$ or $\nx{i}$ until $\ci$ are pebbled, which happens before $\bi$ are pebbled. By normality, all pebbles except the ones on $\nxp{i}$ are removed from the connecting pyramids as soon as $\nxp{i}$ are pebbled. Let $t_2$ be the time the pebbles on the aforementioned pyramids are removed. Let $t_3$ be the first time after $t_2$ that $\qp{i+1}$ is pebbled. 

At $t_2$, there are pebbles on $\xp{i}$, $\di$, and $\nxp{i}$. $K$ pebbles must remain on $\xp{i}$ and $\di$, each, until $t_3$. Furthermore, $K$ pebbles must remain on either $\nxp{i}$ or $\nx{i}$ until $t_3$. First, suppose $K$ pebbles remain on $\nxp{i}$ from $t_2$ to $t_3$. The configuration at $t_2$ is in $N_{i + 1}$ with a double false assignment to $x_i$, and none of the pebbles on the graph at $t_2$ can be removed until $t_3$. By the induction hypothesis, we pebble $\qp{i + 1}$. $Q_{i}$ can subsequently be pebbled with $3K$ pebbles that remained on the block up till $t_3$. Therefore, $\q{i}$ can be pebbled with $s_i$ additional pebbles with moving any of the pebbles in the $N_i$ configuration. Since $Q_{i+1}x_{i+1} \cdots Q_n x_n F(e_1, \cdots, e_{2i-2}, False, False)$ is true then $\forall x_i Q_{i+1}x_{i+1} \cdots Q_n x_n F(e_1, \cdots, e_{2i-2})$ is also true. 

In the case when the $K$ pebbles on $\nxp{i}$ do not remain on $\nxp{i}$ until $t_3$, we argue that $\qp{i+1}$ must be pebbled twice, once with a false assignment to $x_i$ and then with a true assignment to $x_i$. Either $\nxp{i}$ or $\nx{i}$ must be pebbled from $t_2$ to $t_3$. The only successors of $\x{i}$ are $\fii$ and $\fii$ cannot be pebbled before $t_3$. Therefore, by regularity of pebbling strategies, we can arrange to move the $K$ pebbles on $\nxp{i}$ to $\nx{i}$ in the time range $[t_2, t_2 + K^2]$ (and $t_3 > t_2 + K^2$) where they remain until $t_3$. The configuration at $t_2 + K^2$ is that $N_{i + 1}$ is assigned to the false assignment of $x_i$ and none of the pebbles on the graph at $t_2 + K^2$ can be removed until $t_3$ by normality. By the induction hypothesis, $Q_{i + 1}x_{i + 1} \cdots Q_n x_n F(e_1, \cdots, e_{2i - 2}, False, True)$ is true. 

At $t_3$, there are pebbles on $\di$, $\nx{i}$, $\xp{i}$ and $\qp{i + 1}$. Vertices $\q{i}$, $\ai$, $\bi$, $\ci$, $\fii$, and $\gi$ are vacant because they cannot be pebbled before $\qp{i + 1}$ are pebbled. Vertices $\nxp{i}$ couldn't have been repebbled between $t_2 + K^2$ and $t_3$ since $4K$ pebbles are fixed on $\di$, $\nx{i}$, $\xp{i}$, and the paths from $\pii{0}$ to $q_1$ during that interval; thus $\nxp{i}$ and, by normality, the pyramids connected to the vertices cannot be pebbled in the interval. It does not matter whether there exists pebbles on $\x{i}$ at $t_3$. We now show that immediately after $t_3$, a configuration in $N_{i + 1}$ with a true assignment to $x_i$ is created, and that $\qp{i + 1}$ must be repebbled while the pebbles in the configuration are fixed. 

By frugality, the pebbles on $\qp{i +1}$ at $t_3$ remains until either $\ci$ or $\gi$ are pebbled. Vertices $\qp{i + 1}$ cannot contain pebbles until $\gi$ are pebbled since to pebble $\gi$ requires all but $2K$ pebbles on the quantifier block or on the pyramids underneath $\nxp{i}$. $K$ pebbles are fixed on $\xp{i}$, $\x{i}$, or $\fii$ and $K$ pebbles are fixed on $\di$, $\bi$, or $\ai$ until $\q{i}$ is pebbled. Thus, the $K$ pebbles on $\qp{i + 1}$ at $t_3$ remain until $\ci$ are pebbled and are removed before $\gi$ are pebbled. Since $\nx{i}$ are pebbled at $t_3$ we can rearrange the strategy so that the $K$ pebbles from pebbling $\qp{i+1}$ are moved to $\ci$ by time $t_3 + K$. 

Now the only successors of $\ci$ and $\bi$ are $\bi$ and $\ai$, respectively. Since $\di$ and $\xp{i}$ both contain pebbles at $t_3 + K$, we can rearrange the strategy so that the pebbles on $\ci$ are moved to $\bi$ at $t_3 + 2K$ and to $\ai$ at $t_3 + 3K$. $K$ pebbles must remain on $\ai$ until $\q{i}$ are pebbled. Since $\ai$ are only pebbled once after $t_0$ and before $\q{i}$ are pebbled and are the only successors of $\xp{i}$ aside from $\x{i}$, we can further rearrange the strategy so that the pebbles on $\xp{i}$ are moved to $\x{i}$ at $t_3 + 3K + K^2$. 

At $t_3 + 3K + K^2$, $\ai$ contains $K$ pebbles that will remain until $\q{i}$ are pebbled, and $\x{i}$ contain $K$ pebbles that remain until $\fii$ are pebbled. Vertices $\nxp{i}$ must be repebbled before $\fii$ are pebbled, which must happen before $\q{i}$ are pebbled. To pebble $\nxp{i}$, by Lemma~\ref{lem:quant}, requires all the pebbles from this block except the ones on $\ai$ and $\x{i}$, so by normality $\nxp{i}$ are first pebbled after $t_3 + 3K + K^2$, and are each only pebbled once before $\fii$ are pebbled. Let $t_4$ be the time all the pebbles except the ones on $\nxp{i}$ are removed from the pyramids under the nodes where $t_4 > t_3 + 3K + K^2$. At $t_4$, there are pebbles on $\ai$, $\x{i}$, and $\nxp{i}$ and nowhere else on the $i$-th quantifier block. This configuration is in $N_{i + 1}$ with a true assignment to $x_i$ and none of the pebbles on the graph at $t_4$ can be removed until after $\qp{i+1}$ are repebbled. By the induction hypothesis, $Q_{i + 1}x_{i + 1} \cdots Q_n x_n F(e_1, \cdots, e_{2i - 2}, True, False)$ is true. Therefore, $\forall x_i Q_{i +1}x_{i + 1} \cdots Q_{n}x_n F(e_1, \cdots, e_{2i-1})$ is true. This concludes the inductive step for the universal quantifier. 

\item Suppose $Q_i = \exists$. Suppose $t_1$ is a time that $\xp{i}$ is pebbled. By frugality, each of $\ai$, $\bi$, $\ci$, $\di$, $\ei$ and $\qp{i + 1}$ are pebbled at most once after $t_1$ and before $\q{i}$ are pebbled. Exactly as in the case of the universal quantifier, normality implies that $\xp{i}$ are only pebbled once after $t_1$ and before $\q{i}$ are pebbled, and are pebbled before anything else happens. $K$ pebbles remain on $\xp{i}$ or $\x{i}$ until $\ai$ are pebbled, and $K$ pebbles remain on $\xp{i}$ or $\nxp{i}$ until $\bi$ are pebbled. To pebble $\ai$ require pebbling $\di$, which require removing all pebbles from the block except the ones on $\xp{i}$. Thus, $\di$ are pebbled before anything else except $\xp{i}$, and pebbles remain on $\di$ until $\ci$ are pebbled. 

To pebble $\ci$ require pebbling $\ei$ and hence $\fii$. To pebble $\fii$ requires removing all pebbles from this block except those on $\xp{i}$ and $\di$. Thus, $\fii$ are pebbled only once before $\ei$ are pebbled, and this happens immediately after $\di$ are pebbled. $K$ pebbles remain on $\fii$, $\nxp{i}$, or $\nx{i}$ until $\ei$ are pebbled. The only successors of $\fii$ are $\nxp{i}$, and $K$ pebbles remain on $\xp{i}$ until $\nxp{i}$ are pebbled, so we can rearrange the strategy so that the first move after picking up the pebbles on the pyramids underneath $\fii$ is to move the pebbles on $\fii$ to $\nxp{i}$. Let $t_2$ be the time of this move, and let $t_3$ be the time $\qp{i+1}$ are pebbled. Note that since $\fii$ are not repebbled between $t_2$ and $t_3$, neither are $\nxp{i}$. At $t_2$, there are pebbles on $\xp{i}$, $\nxp{i}$, and $\di$ and until $t_3$, there must be pebbles on $\xp{i}$ or $\x{i}$, $\xp{i}$ or $\nxp{i}$, $\nxp{i}$ or $\nx{i}$, and $\di$. 

Now we consider $3$ cases. Suppose that the pebbles on $\xp{i}$ are removed before $t_3$. Since the only successors of $\xp{i}$ are $\x{i}$ and $\nxp{i}$ and $\nxp{i}$ is not repebbled before $t_3$, we can rearrange the strategy so that the pebbles on $\xp{i}$ are moved to $\x{i}$ at $t_2 + K^2$. The configuration at $t_2 + K^2$ is then in $N_{i + 1}$ with the true assignment to $x_i$, and none of the pebbles can be removed until $t_3$. By the induction hypothesis, $Q_{i + 1}x_{i + 1} \cdots Q_n x_n F(e_1, \cdots, e_{2i - 2}, True, False)$ is true.

Suppose in the second case that $K$ pebbles remain on $\xp{i}$ until $t_3$, and the pebbles on $\nxp{i}$ are removed before $t_3$. We can rearrange the strategy so that the pebbles on $\nxp{i}$ are moved to $\nx{i}$ at $t_2 + K^2$. The configuration at $t_2 + K^2$ is in $N_{i + 1}$ with the false assignment to $x_i$, and no pebble can be removed until $t_3$. By the induction hypothesis, $Q_{i + 1} x_{i +1} \cdots Q_n x_n F(e_1, \cdots, e_{2i -2}, False, True)$ is true.

Finally, suppose that pebbles remain on $\xp{i}$ and $\nxp{i}$ until $t_3$. The configuration at $t_2$ is in $N_{i + 1}$ with a double false assignment to $x_i$, and no pebble is removed until $t_3$. By the induction hypothesis, $Q_{i + 1}x_{i+1} \cdots Q_n x_n F(e_1, \cdots, e_{2i - 2}, False, False)$ is true. 

In each of the above cases, $\exists x_i Q_{i+1}x_{i+1} \cdots Q_n x_n F(e_1, \cdots, e_{2i-2})$ is true. This completes the inductive step for an existential quantifier, and the proof of the claim. 
\end{enumerate}
\end{proof}

The proof of the above claim subsequently proves the lemma when $i = 1$. 
\end{proof}

We now prove that $3Ku + 5K$ pebbles are necessary to pebble an unsatisfiable instance of QBF. 

\begin{lemma}\label{lem:lower-bound-false}
Given $\sg{2}$ which is constructed from the provided QBF instance, $B = Q_1x_1 \cdots Q_u x_u F$, using our modified reduction in~\Cref{sec:modified} and~\Cref{sec:full-construction}, $B$ is unsatisfiable if and only if $\p(\sg{2}) \geq 3Ku + 5K$.
\end{lemma}

\begin{proof}
We first prove that if $B$ is unsatisfiable, then the number of pebbles necessary to pebble the modified construction requires at least $3Ku + 5K$ pebbles.
By Lemma~\ref{lem:quant}, there does not exist a frugal strategy such that $3Ku$ pebbles are not assigned to the $u$ quantifier blocks when the clauses are pebbled and the number of pebbles used is less than $3Ku + 5K$. Therefore, there exists only $s - 3Ku \leq 4K + 1$ pebbles remaining to pebble the clauses assuming the player is given $3Ku + 4K + 1$ to begin with to pebble $\sg{2}$. 
Provided $3Ku$ pebbles are on the quantifier blocks, Lemma~\ref{lem:clause} proves that at least $5K$ additional pebbles are needed to pebble one or more false clauses given that $B$ is unsatisfiable.  

Now we prove that if the number of pebbles necessary to pebble $\sg{2}$ is at least $3Ku + 5K$, then $B$ is unsatisfiable. This proof is given by contradiction which immediately follows from Lemma~\ref{lem:upper-bound-true}. 
\end{proof}

\subsubsection{Proof of Inapproximability}\label{sec:inapprox-proof}

Using Lemmas~\ref{lem:upper-bound-true} and~\ref{lem:lower-bound-false}, we prove that it is PSPACE-hard to approximate the minimum number of black pebbles needed given a DAG, $\sg{2}$, to an additive $n^{1/3-\epsilon}$ for all $\epsilon > 0$.  

\begin{lemma}\label{lem:total-nodes}
The number of nodes in $G$ is $O(K^3(u^3 + c))$.
\end{lemma} 

\begin{proof}
By construction of $G$ as defined in Sections~\ref{sec:modified} and~\ref{sec:full-construction}, we create $u$ variable gadgets, $u$ quantifier blocks, and $c$ clause gadgets. Each variable gadget contains $O(K^3)$ nodes since it contains two road graphs where each road graph with $K$ width contains $K^3$ nodes. Each quantifier block contains a variable gadget and the pyramids that connect to the variable gadget and $O(K)$ other nodes. The total size of the pyramids is at most $\sum_{i = 1}^{3Ku + 4K + 1} i^2 = O(K^3u^3)$. Therefore, the total size of all the quantifier blocks is $O(K^3 u^3)$ since the total size of the quantifier blocks without the connecting pyramids is $O(K^3 u)$ and the total size of all the pyramids is $O(K^3 u^3)$. 

The clauses each have size $O(K^2)$ since the clauses solely consist of a constant number of pyramids of height $O(K)$. Therefore, the total size of all the clauses is $O(K^2 c)$. 

Thus, $G$ has $O(K^3(u^3 + c))$ number of nodes in total. 
\end{proof}

\begin{theorem}[Restatement of Theorem~\ref{thm:approx-standard}]
The minimum number of pebbles needed in the standard pebble game on DAGs with maximum indegree $2$ is PSPACE-hard to approximate to additive $n^{1/3-\epsilon}$ for $\epsilon > 0$.
\end{theorem}

\begin{proof}
From Lemmas~\ref{lem:upper-bound-true} and~\ref{lem:lower-bound-false}, the cost of pebbling a graph constructed from a satisfiable $B$ is at most $3Ku + 4K+ 1$ whereas the cost of pebbling a graph constructed from an unsatisfiable $B$ is at least $3Ku + 5K$. 

As we can see, the aforementioned reduction is a gap-producing reduction with a gap of $K-1$ pebbles. Then, all that remains to be shown is that for any $\epsilon > 0$, it is the case that $K \geq (K^3(u^3 + c))^{(1/3-\epsilon)}$. (Note that for $\epsilon > 1/3$, setting $K$ to any positive integer achieves this bound.) Suppose we set $K = \max(u, c)^a$ where $a > 0$. We show that $K = (K^3 [\max(u, c)]^3)^{(1/3 - \epsilon)}$ for some valid setting of $a$ for every $0 < \epsilon \leq 1/3$. Solving for $a$ in terms of $\epsilon$ gives us $a = \frac{1}{3\epsilon} - 1 \geq 0$ when $\epsilon \leq 1/3$ and is finite when $\epsilon > 0$. 

For values of $a \geq 0$, we can duplicate the clauses and variables gadgets so that $u$ and $c$ are large enough such that $K = \max(u, c)^a \geq 2$. Let $d = \max(u, c)$. Then, we need $d$ to be large enough so that $d^a \geq 2$ (i.e. we want $d^a$ to be some integer). Then, we can set $d \geq 2^{1/a}$. Thus, we can duplicate the number of variables and clauses so that $d \geq 2^{\frac{3\epsilon}{1-3\epsilon}}$. (Note that for cases when $a$ is very small, e.g. $a = 0$ when $\epsilon = 1/3$, any constant $K$ would suffice.)

Therefore, for every $\epsilon > 0$, we can construct a graph with a specific $K$ calculated from $\epsilon$ such that it is PSPACE-hard to find an approximation within an additive $n^{1/3-\epsilon}$ where $n$ is the number of nodes in the graph.  
\end{proof}

\section{Hard to Pebble Graphs for Constant $k$ Pebbles}     
\label{sec:hard-to-pebble}
It is long known that the maximum number of moves necessary to pebble any graph with constant $k$ pebbles is $O(n^k)$. (Note that the maximum number of moves necessary to pebble any graph is either $O(n^{k-1})$ or $O(n^k)$ depending on whether or not sliding is allowed. Here, we allow sliding in all of our games. The bound of $O(n^{k-1})$ proven in~\cite{Nor15} is one for the case when sliding is not allowed.) The upper bound of $O(n^k)$ for any constant $k$ number of pebbles submits to a simple combinatorial proof adapted from \cite{Nor15} to account for sliding. However, to the best of the author's knowledge, examples of such families of graphs that require $O(n^k)$ moves to pebble using $k$ pebbles did not exist until very recently in an independent work~\cite{ADNV17}. In this section, we present an independent, simple to construct family of graphs
that require $\Theta(n^k)$ time for constant $k$ number of pebbles in both the standard and black-white pebble games. We further reduce
the indegree of nodes in this family of graphs to $2$ and show that our results
still hold. Furthermore, we show this family of graphs to exhibit a steep time-space trade-off (from exponential in $k$ to linear) even when $k$ is not constant. 
Such families of graphs could potentially have useful applications in cryptography in the domain of proofs of space and memory-hard functions~\cite{AlwenS15}.

We construct the following family of graphs, $\mathbb{H}$, below and show that for constant $k$ pebbles, the number of steps it takes to pebble the graph
$H_{n, k} \in \mathbb{H}$ with $k$ pebbles and $n$ nodes is $\Omega(n^k)$. We also show a family of graphs, $\shf{2}$ with indegree $2$ that exhibits the same asymptotic tradeoff. 

We construct the family of graphs $\mathbb{H}$ with arbitrary indegree in the following way. 

\begin{definition}\label{def:near-optimal-tradeoff}
Given a set of $n$ nodes and maximum number of pebbles $k$ where $k < \sqrt{n}$, we lexicographically order the nodes (from $1$ to $n$) and create the following set of edges between the nodes where directed edges are directed from $v_i$ to $v_j$ where $i < j$. Let $[n]$ be the ordered set $[1, \dots, n]$:
\begin{enumerate}
\item $v_i$ and $v_{i+1}$ for all $i \in [k-1, n]$ 
\item $v_i$ and $v_j$ for all $i \in [l-1]$ for all $2 \leq l \leq k$ and $j \in \left\{f(l) + 2r - 2\right\}$ for all $r \in [\frac{n-k}{2k}]$ where $f(l) = (k-1) + (l-1)(\frac{n-k}{k}) + 2$.
\item $v_i$ and $v_j$ for all $i = f(l) - 2$ and $j \in \left\{f(l) + 2r - 1\right\}$ for all $r \in [0, \frac{n-k}{2k} - 1]$ where $l \in [1, k-1]$.
\end{enumerate}

The target node (the only sink) is $v_{n}$. Note that the sources in our construction are $v_j$ for all $j \in [1, k-1]$.
\end{definition}

Below is an example graph (Fig.~\ref{fig:hard-to-pebble-example}) in our family when $k = 5$ and $n = 54$. 

\begin{figure}[h]
\centering
\def\svgwidth{\columnwidth}
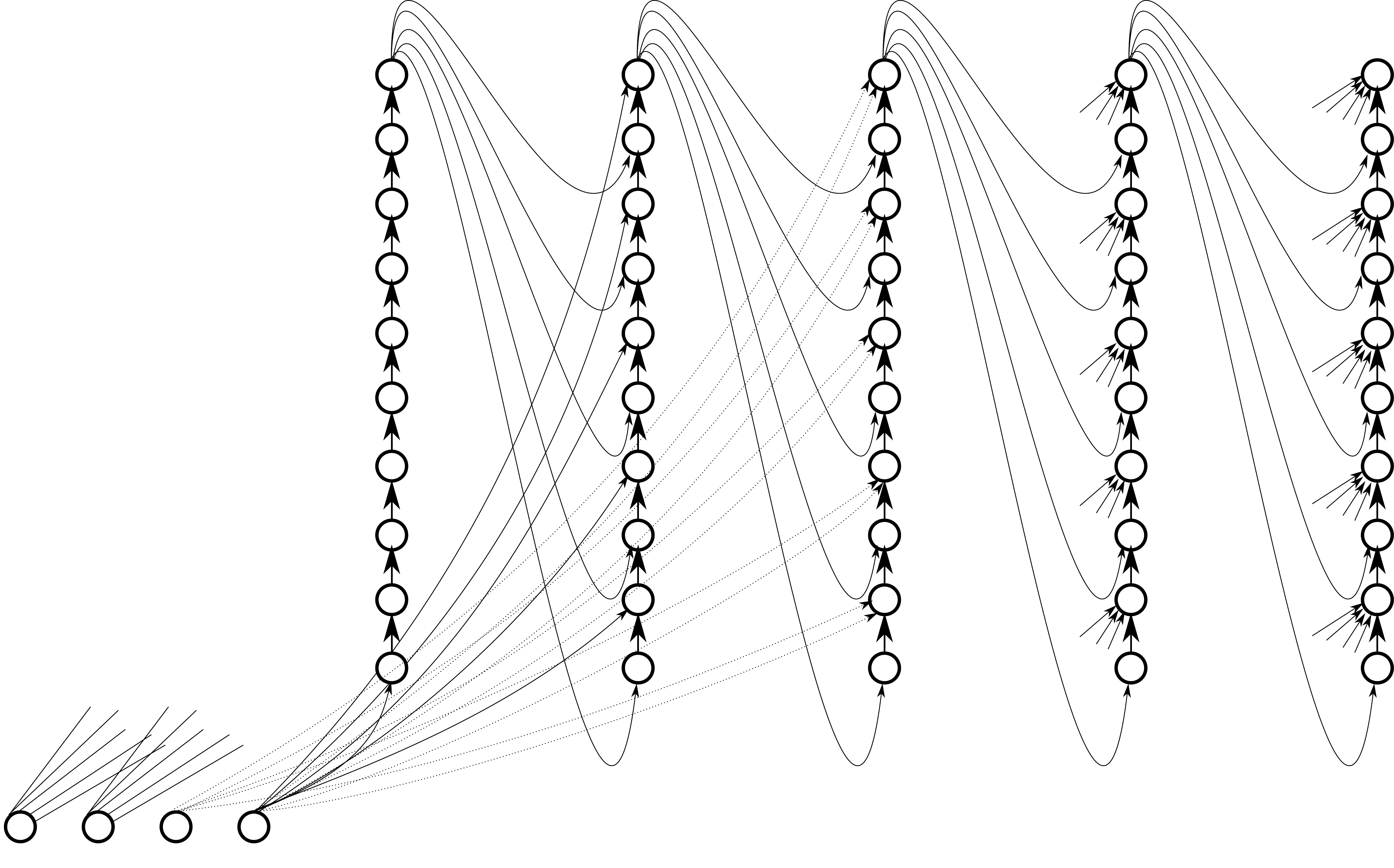
\caption{Example member of class of graphs where $k = 5$ and $n = 54$. One can make any graph of this class into an indegree-$2$ graph by replacing the input vertices by pyramids and performing the modifications described in Section~\ref{sec:hard-modifications}.}\label{fig:hard-to-pebble-example}
\end{figure}

We now prove the time bound for this family of graphs $\hf$ for all $k < \sqrt{n}$.

To prove the minimum number of pebbles necessary to pebble the graph, it is sufficient to study the number of blocked paths in any graph $G$~\cite{Nor15}. We define a \emph{blocked path} as in~\cite{Nor15}.

\begin{definition}[Blocked Paths \cite{Nor15}]\label{def:blocking}
A set of vertices, $U$, blocks a path, $P$, if $U \cap P \neq \emptyset$. $U$ blocks a set of paths $\mathbb{P}$ if $U$ blocks $P$ for all $P \in \mathbb{P}$.
\end{definition}

\begin{lemma}\label{lem:min-k-unbounded}
The minimum number of pebbles necessary to pebble $\sh{k} \in \hf$ is $k$.
\end{lemma}

\begin{proof}
The degree of the graph $\sh{k}$ is $k$, therefore, at least $k$ pebbles are necessary to pebble $\sh{k} \in \hf$.
\end{proof}

\begin{theorem}\label{thm:nk-moves-bound}
The number of moves necessary to pebble $\sh{k}$ is $\Theta((\frac{n-k}{2k})^k)$ for $k < \sqrt{n}$.
\end{theorem}

\begin{proof}
Let $\s = \{P_0, \dots, P_{a}\}$ where $a = \Theta((\frac{n-k}{2k})^k)$. Suppose that at time $t$, there exists at least $k$ paths from sources to targets which are blocked by $k$ pebbles placed on the graph. By Lemma~\ref{lem:min-k-unbounded} and by construction of $\sh{k}$, we know $t$ must exist at some point in the pebbling of $\sh{k}$. Let $v_{i}$ be a degree-$k$ node to be pebbled. Let $t$ be the time step immediately before $v_i$ is pebbled. Then, $|P_t| = k$ and there exists $k$ blocked paths from each $u \in \pred{v_i}$ that is a source and $v_k$ to the target $v_n$. Each of the paths is blocked by a pebble on $\pred{v_i}$ at time $t$. 

Let $t_{now}  = t + 1$ be the time when $v_{i}$ is pebbled. After $v_i$ has been pebbled, there are two different paths going from the sources through $v_{i+1}$. However, none of the pebbles placed at time $t' \in [t, t_{now}]$ blocks the path from $v_k$ to $v_{n}$ that does not include $v_{i}$. Therefore, a pebble must be placed on this path to block it, resulting in repebbling $v_{(k-1) + (k-1)(\frac{n-k}{k})}$.

By induction on the level number, $1 \leq l \leq k$, where $T(1) = \frac{n-k}{k}$ is the base case when node $v_{(k-1) + \frac{n-k}{k}}$ is pebbled. We see by the above argument that the resulting number of moves is $T(l) = (\frac{n-k'}{2k'}) T(l-1) + l(\frac{n-k'}{2k'})$ where $k' = k$ is not a variable in the recursion (i.e. not changing) and we compute $T(k)$.  Therefore, 

\begin{align*}
T(k) = (\frac{n-k}{2k})^k + \sum_{i=0}^{k-1}(k-i)(\frac{n-k}{2k})^{i+1} \geq 2(\frac{n-k}{2k})^k= \Theta((\frac{n-k}{2k})^k).
\end{align*}
\end{proof}

We obtain the following when $k$ is constant:

\begin{corollary}[$\sh{k}$: $\Omega(n^k)$ moves bound for constant $k$]\label{cor:constant-k}
When $k$ is constant, pebbling $\sh{k} \in \hf$ using $k$ pebbles takes $\Omega(n^k)$ time.
\end{corollary}

This result of itself partially answers a longstanding open question posed in~\cite{Nor15} whether a family of graphs with constant degree can have a number of moves, $\Omega(n^k)$, that meets the upper bound for constant $k$ number of pebbles. Now, we completely resolve this open question by proving that the bound also holds for the black-white pebble game using our independent construction from~\cite{ADNV17}.

Before we prove the result, we quickly give the rules for the black-white pebble game as promised in Section~\ref{sec:intro}.

\textbf{Black-White Pebble Game Rules:}\\\\
\noindent\fbox{%
    \parbox{\textwidth}{%
    \textsc{Black-White Pebble Game}
    
    \textbf{Input:} Given a DAG, $G = (V, E)$. Let $\pred{v} = \{u \in V: (u, v) \in E \}$. Let $S \subseteq V$ be the set of sources of $G$ and $T \subseteq V$ be the set of targets of $G$. Let $\s = \{P_0, \dots, P_{\tau}\}$, where $P_i = (B_i, W_i)$ ($B_t$ is the set of nodes with black pebbles and $W_t$ is the set of nodes with white pebbles), be a valid pebbling strategy that obeys the following rules where $P_0 = (\emptyset, \emptyset)$ and $P_{\tau} = (T, \emptyset)$. Let $\p(G, \s) = \max_{i \in [\tau]}\{|P_i|\}$ where $|P_i| = |B_i| + |W_i|$.\\
    
    \textbf{Rules:}
	\begin{enumerate}
		\item At most one pebble can be placed or removed from a node at a time.
		\item A black pebble can be placed on any source, $s \in S$. A white pebble can always be removed from a source.
		\item A black pebble can be removed from any vertex. A white pebble can be removed from a non-source vertex, $v$, at time $i$ if and only if $\pred{v} \in P_{i-1}$. 
		\item A black pebble can be placed on a non-source vertex, $v$, at time $i$ if and only if $\pred{v} \in P_{i-1}$. A white pebble can be placed on an empty vertex at any time.
		\item A black pebble can be moved from vertex $v$ to vertex $w$ at time $i$ if and only if $(v, w) \in E$ and $\pred{w} \in P_{i - 1}$. A white pebble can be moved from vertex $w$ to vertex $v$ if and only if $(v, w) \in E$ and $\pred{w} \backslash v \in P_{i-1}$ and $v \not\in P_{i-1}$.
	\end{enumerate}
	
	\textbf{Goal:} Determine $\min_{\s}\{\p(\g, \s)\}$ using a valid strategy $\s$.
	}%
}%

\begin{theorem}[$\sh{k}$: $\Omega(n^k)$ moves bound black-white pebble game]\label{thm:no-moves-black-white}
The number of moves necessary to pebble $\sh{k}$ is $\Theta((\frac{n-k}{2k})^k)$ for $k < \sqrt{n}$ using the rules of the black-white pebble game.
\end{theorem}

\begin{proof}
Let $\s = \{P_0, \dots, P_a\}$ be a black-white pebbling strategy where $a = \Theta((\frac{n-k}{2k})^k)$. At least $k$ pebbles must be used to pebble nodes with indegree $k$. We proved in Theorem~\ref{thm:nk-moves-bound} that strategies that use only black pebbles  must make $\Theta((\frac{n-k}{2k})^k)$ moves. Therefore, in order to use fewer than $\Theta((\frac{n-k}{2k})^k)$ moves, at least one white pebble must be used. 

Let $\tbw(n, k)$ be the minimum time of pebbling $\shn{k}{n}$ using $k$ black and white pebbles. $\tbw(\frac{n-k}{k} + 1, 1) \geq \frac{n-k}{k} + 1$. Because of the recursive structure of the graph family, we now show that 

\begin{align*}
\tbw(n, l) \geq (\frac{n-k}{2k})\tbw((l-1)(\frac{n-k}{k}) + l - 1, l-1) + l(\frac{n-k}{2k})
\end{align*}

even when using black and white pebbles. Solving for $T(k)$ gives the number of moves to pebble $H_{n, k}$ using the rules of the black-white pebble game. The proof of the theorem then follows directly from the base case and the proof of Theorem~\ref{thm:nk-moves-bound}. 

To show the above, we first consider the case when $v_n$ is pebbled with a white pebble. In this case, the predecessors of $v_n$ must be pebbled with either black or white pebbles. $v_n$ must be repebbled with a black pebble after the predecessors are pebbled, resulting in a strategy that does not use the minimum number of moves. Therefore, $v_n$ is never pebbled with a white pebble. Now consider $\pred{v_n}$. If $v_n$ has $k$ predecessors, suppose the non-source predecessor of $v_n$ is pebbled at $t_1$. Then, the earliest that $v_n$ can be pebbled is at time $t_1+ k - 1$ regardless whether or not black or white pebbles are used. Given that $k$ pebbles must be used to pebble $v_{n-i}$ where $i$ is even and $v_{n-i}$ has $k$ predecessors, $v_{k + (k-1)(\frac{n-k}{k})}$ must be pebbled $\frac{n-k}{2k}$ times. Thus, we have shown that regardless of whether black or white pebbles are used, we reduce to the above recursive relation.
\end{proof}

\subsection{Max Indegree-$2$ Hard to Pebble Graphs}\label{sec:hard-modifications}

If we modify the construction presented in Definition~\ref{def:near-optimal-tradeoff} such that every node of degree $d > 2$ is replaced with a pyramid of height $d$, then we obtain the results we would like for the standard pebble game taking $\Omega(n^k)$ moves for any $n$ and constant $k$. Rather than creating a unique pyramid for each node of degree $d > 2$, we create one height $k$ pyramid $\Pi_h$ and connect it to our construction, described below. From this construction, we obtain Theorem~\ref{thm:almost-maximum-time} as stated in the introduction.

\begin{definition}[Standard Pebbling Construction with Max Indegree-$2$]\label{def:stand-max-deg}

We create the max indegree-$2$ hard to pebble family of graphs using the standard pebble game as follows. Suppose we have a total of $n$ nodes. 

\begin{enumerate}
\item Create a height $k-1$ pyramid and label the roots of pyramids of heights in the range $i \in [1, k-1]$, $r_i$. 
\item Sort the remaining $n - \frac{(k-1)k}{2}$ vertices and create edges $(v_i, v_j)$ where $i < j$ in the sorted order.
\item Create edges $(r_i, v_j)$ for all $i \in [1, l-1]$ for all $1 \leq l \leq k$ and $j = f(l) + i + g + 1$ for all $g \in [\frac{n - \frac{(k-1)k}{2}}{kl}]$ and $f(l) = \frac{(k-1)(k)}{2} + (l-1)(\frac{n-\frac{(k-1)k}{2}}{k})$. 
\item Create edges $(v_i, v_j)$ for all $i = f(l) - 1$ and $j \in \left\{f(l) + gl\right\}$ for all $g \in [0, \frac{n - \frac{(k-1)k}{2}}{kl}]$ where $l \in [1, k-1]$.
\end{enumerate}

The target node is $v_n$. 
\end{definition}

\begin{theorem}\label{thm:b-almost-maximum-time}
There exists a family of graphs with $n$ vertices and maximum indegree $2$ such that $\Omega((\frac{n-k^2}{k^2})^k)$ moves are necessary to pebble any graph with $n$ vertices in the family using $k < \sqrt{\frac{n}{2}}$ pebbles in the standard pebble game.
\end{theorem}

\begin{proof}
By Definition~\ref{def:stand-max-deg}, creating the pyramid of height $k-1$ requires $O(k^2)$ nodes. It suffices to only create one pyramid since a pyramid of any height less than $k-1$ is contained in a pyramid of height $k-1$. Furthermore, considering that the same $k$ nodes are used for all nodes in different columns in Theorems~\ref{thm:nk-moves-bound} and~\ref{thm:no-moves-black-white}, it does not matter that the different pyramids share nodes. For every node of degree $d \geq 3$ in the construction defined by Definition~\ref{def:near-optimal-tradeoff}, we replace the node by a path with nodes connected to pyramids of heights, $h \in [1, d-1]$. By normality, for any pyramid of height $h$, we must remove $h$ pebbles from the graph to pebble it in the standard pebble game. For paths that connect to pyramids of height $h \in [1, k-1]$, there are only two ways to pebble the pyramid of height $i$. Either a pebble remains on the apex of the pyramid of height $i$ or $i$ pebbles are removed from the graph to pebble the pyramid using a normal strategy. To pebble the path connected to pyramids of heights $h \in [1, k-1]$ requires a total of $k-1$ pebbles either remaining on the pyramids or removed from the graph to be used to pebble the apex of each of the pyramids. 

Suppose $T(l)$ is the time of pebbling the last node in the topological order of layer $l$ with base case $T(1) =  \frac{n - \frac{(k-1)k}{2}}{k}$. Then, we obtain the following recursive equation from our construction in Definition~\ref{def:stand-max-deg}:

\begin{align*}
T(l) = \frac{n - \frac{(k-1)k}{2}}{kl}T(l-1) + (l-1)\Theta(\frac{n - k^2}{kl}) = \Omega((\frac{n-k^2}{k^2})^k)
\end{align*}
\end{proof}

The proof of number of standard pebbling moves necessary in pebbling the family of graphs defined by Definition~\ref{def:stand-max-deg} is $\Omega(n^k)$ when $k$ is constant, proving part of Theorem~\ref{thm:almost-maximum-time}.

To prove the result for the black-white pebble game, we use a result from~\cite{L79} and~\cite{Nor15} that gives a precise space cost for a complete binary search tree. 

\begin{theorem}[Black-White Pyramid Pebble Price~\cite{L79, Nor15}]\label{thm:bw-pyramid-pebble-price}
For a complete binary tree $T_h$ of height $h \geq 1$ it holds that the black-white persistent pebbling cost is $\floor{\frac{h + 3}{2}}$. The \emph{persistent} pebbling cost is defined as the cost of pebbling the root of $T_h$ with a black pebble that remains on the root. 
\end{theorem}

We state as an immediate corollary of Theorem~\ref{thm:bw-pyramid-pebble-price}:

\begin{corollary}\label{cor:white-pebble-root}
For a complete binary tree $T_h$ of height $h \geq 1$ where $h \mod 2 = 1$, the cost of pebbling the root of $T_h$ in the first step using a white pebble is $\floor{\frac{h + 3}{2}}$. 
\end{corollary}

\begin{proof}
Suppose for the sake of contradiction that the cost of pebbling the root of $T_h$ in the first step using a white pebble is less than $\floor{\frac{h+3}{2}}$, then according to the algorithm presented in~\cite{L79}, the cost of persistently pebbling a binary tree of height $h'$ where $h' = h + 1$ is equal to $\floor{\frac{h' + 3}{2}}$ since the original strategy of pebbling one predecessor with a black pebble persistently and the other one with a white pebble results in the persistent cost to be $\floor{\frac{h'+3}{2}}$. However, this contradicts with the stated lower bound of $\ceil{\frac{h'+3}{2}} = \floor{\frac{h'+3}{2}} + 1$~\cite{L79} since $h' \mod 2 = 0$.  
\end{proof}

Using Theorem~\ref{thm:bw-pyramid-pebble-price} and Corollary~\ref{cor:white-pebble-root}, we can define a class of graphs very similar to the class of graphs defined by Definition~\ref{def:stand-max-deg}. 

\begin{definition}[Black-White Pebbling Construction with Max Indegree-$2$]\label{def:bw-max-time}
We create the max indegree-$2$ hard to pebble family of graphs using the black-white pebble game as follows. Suppose we have a total of $n$ nodes. 

\begin{enumerate}
\item Create a height $H = 2(k-1)-3 = 2k - 5$ complete binary tree and label the roots of trees of heights $2i - 3$ for $i \in [1, k-1]$, $r_i$. 
\item Sort the remaining $n - 2^{2k-5}$ vertices and create edges $(v_i, v_j)$ where $i < j$ in the sorted order.
\item Create edges $(r_i, v_j)$ for all $i \in [1, l-1]$ for all $1 \leq l \leq k$ and $j = f(l) + i + g + 1$ for all $g \in [\frac{n - 2^{2k-5}}{kl}]$ and $f(l) = 2^{2k-5} + (l-1)(\frac{n-2^{2k-5}}{k})$. 
\item Create edges $(v_i, v_j)$ for all $i = f(l) - 1$ and $j \in \left\{f(l) + gl\right\}$ for all $g \in [0, \frac{n - 2^{2k-5}}{kl}]$ where $l \in [1, k-1]$.
\end{enumerate}

The target node is $v_n$. 
\end{definition}

Now we can prove our main theorem for black-white pebbling of our modified graph class as defined in Definition~\ref{def:bw-max-time}.

\begin{theorem}\label{thm:bw-almost-maximum-time}
There exists a family of graphs with $n$ vertices and maximum indegree $2$ such that $\Omega((\frac{n-2^{(2k - 5)}}{k^2})^k)$ moves are necessary to pebble any graph with $n$ vertices in the family using $k = o(\log{n})$ pebbles in the black-white pebble game.
\end{theorem}

\begin{proof}
We create one complete binary tree of height $2k - 5$. For every node of degree $d > 2$, we create a path where each node in the path is connected to roots of trees of heights $2i - 3$ for all $i \in [1, d-1]$. As in the proof for Theorem~\ref{thm:b-almost-maximum-time}, $d-1$ pebbles in total must either be on the roots of the trees or removed from the graph to pebble the roots of these trees regardless of whether black or white pebbles are used (by Theorem~\ref{thm:bw-pyramid-pebble-price} and Corollary~\ref{cor:white-pebble-root}).

This will ensure that all $d$ pebbles are used to pebble the binary search trees and the last node of the path. Therefore, we reach a similar recursive equation as in Theorem~\ref{thm:b-almost-maximum-time}, using $T(l)$ as the time cost of pebbling level $l$ with base case $T(1) = \frac{n-2^{2k-5}}{k}$:

\begin{align*}
T(l) = \frac{n - 2^{2k-5}}{kl}T(l-1) + (l-1)\Theta(\frac{n -2^{2k-5}}{kl}) = \Omega((\frac{n-2^{2k-5}}{k^2})^k)
\end{align*}
\end{proof}

The proof of number of black-white pebbling moves necessary in pebbling the family of graphs defined by Definition~\ref{def:bw-max-time} is $\Omega(n^k)$ when $k$ is constant, concluding the proof of Theorem~\ref{thm:almost-maximum-time}.

\section{Open Problems}
\label{sec:open-problems}
There are a number of open questions that naturally follow the content of this paper. 

The first obvious open question is whether the techniques introduced in this paper can be tweaked to allow for a PSPACE-hardness of approximation to an $n^{1-\epsilon}$ additive factor for any $\epsilon > 0$. We note that the trivial method of attempting to reduce the size of the subgraph gadgets used in the variables (i.e. use a different construction than the road graph such that less than $K^3$ nodes are used) is not sufficient since the number of nodes in the graph is still $\Theta(K^3(u^3 + c))$. This is not to say that such an approach is not possible; simply that more changes need to be made to all of the other gadgets. The next logical step is to determine whether $\p(\sg{2})$ can be approximated to a constant $2$ factor multiplicative approximation.

Another open question is whether the techniques introduced in this paper can be applied to show hardness of approximation results for other pebble games such as the black-white or reversible pebble games. The main open question in the topic of hardness of approximation of pebble games is whether these pebble games can be approximated to any multiplicative factors smaller than $n/\log{n}$ or whether the games are PSPACE-hard to approximate to any constant factor, perhaps even logarithmic factors. 

With regard to hard to pebble graphs, we wonder if our graph family could be improved to show $\Omega(n^k)$ for any $0< k \leq n/\log{n}$. This would be interesting because to the best of the authors' knowledge we do not yet know of any graph families that exhibit sharp (asymptotically tight) time-space trade-offs for this entire range of pebble number.

We also reiterate the persistent black-white pebbling cost of a pyramid (an open problem presented in~\cite{Nor15}) is an interesting open problem with respective to our results because it would broaden the range of allowed $k$ in Theorem~\ref{thm:bw-almost-maximum-time}. 

\clearpage
\bibliography{ref}



\end{document}